\documentclass[11pt,english]{article}

\usepackage[usenames,dvipsnames,svgnames,table]{xcolor}
\usepackage{amssymb}
\usepackage{amsmath}
\usepackage{amsthm}
\usepackage{todonotes}
\usepackage{bm}
\usepackage{graphicx}
\usepackage{nicefrac}
\usepackage{hyperref}
\usepackage{enumitem}

\long\def\comment #1\commentend{}
\long\def\commentt #1\commenttend{}

\newtheorem{theorem}{Theorem}[section]
\newtheorem{lemma}[theorem]{Lemma}

\newtheorem{corollary}[theorem]{Corollary}
\newtheorem{proposition}[theorem]{Proposition}

\newtheorem{observation}[theorem]{Observation}

\def\Prb{\mbox{\textup{Pr}}}

\def\a{\alpha}
\def\b{\beta}
\def\r{\lambda}
\def\X{\bm{x}}
\def\Xnr{\X^{n,\r}}

\def\F{\mathcal{F}}
\def\H{\mathcal{H}}
\def\M{\mathcal{M}}

\def\Ebb{\mathbb{E}}

\def\expD{\mbox{Exp}}

\def\pareto{\mbox{Pareto}}
\def\fthotel{\mbox{\texttt{FPH}}}

\def\s{\mathbf{s}}
\def\rangeL{B^L}
\def\rangeR{B^R}
\def\utilL{u^L}
\def\utilR{u^R}
\def\TutilL{U^L}
\def\TutilR{U^R}
\def\Sc{\s^{n,f}}

\newcommand{\Sm}[1]{\s_{-#1}}

\newcommand{\Xm}[1]{\X_{-#1}}

\makeatletter
\def\thm@space@setup{%
  \thm@preskip=\parskip \thm@postskip=0pt
}
\makeatother

\title{Hotelling Games with Random Tolerance Intervals}

\author{
Avi Cohen\thanks{Weizmann Institute of Science, Rehovot, Israel.
{\tt \{avi.cohen,david.peleg\}@weizmann.ac.il}.}
\and
David Peleg$^*$
}

\title{Hotelling Games with Random Tolerance Intervals}

\begin{document}

\maketitle

\begin{abstract}
The classical Hotelling game is played on a line segment whose points represent
uniformly distributed clients. The $n$ players of the game are servers
who need to place themselves on the line segment, and once this is done,
each client gets served by the player closest to it. The goal of each player
is to choose its location so as to maximize the number of clients it attracts.

In this paper we study a variant of the Hotelling game where each client $v$ has a {\em tolerance interval}, randomly distributed according to some density function $f$, and $v$ gets served by the nearest among the players {\em eligible} for it, namely, those that fall within its interval. (If no such player exists, then $v$ abstains.)
It turns out that this modification significantly changes the behavior of the game and its states of equilibria. In particular, it may serve to explain why players sometimes prefer to ``spread out,'' rather than to cluster together as dictated by the classical Hotelling game.

We consider two variants of the game: \emph{symmetric} games, where clients have the same tolerance range to their left and right, and \emph{asymmetric} games, where the left and right ranges of each client are determined independently of each other.
We characterize the Nash equilibria of the 2-player game. For $n\geq3$ players, we characterize a specific class of strategy profiles, referred to as {\em canonical profiles}, and show that these profiles are the only ones that may yield Nash equilibria in our game. Moreover, the canonical profile, if exists, is uniquely defined for every $n$ and $f$.
In the symmetric setting, we give simple conditions for the canonical profile to be a Nash equilibrium, and demonstrate their application for several distributions. In the asymmetric setting, the conditions for equilibria are more complex; still, we derive a full characterization for the Nash equilibria of the exponential distribution. Finally, we show that for some distributions the simple conditions given for the symmetric setting are sufficient also for a Nash equilibrium in the asymmetric setting.
\end{abstract}


\section{Introduction}

\subsection{Background and Motivation}
The Hotelling game, introduced in the seminal~\cite{hotelling1929stability}, is a widely studied model of spatial competition in a variety of contexts, ranging from the placement of commercial facilities, to the differentiation between similar products of competing brands,
to the positioning of candidates in political elections. The well known toy example is as follows: two ice cream vendors choose a location on a beach strip. Beach goers are uniformly distributed on the beach, and each buys ice cream from the closest vendor. The goal of each vendor is to maximize the number of customers he
\footnote{In the introduction, we use the noun ``he'' for players, and ``she'' for clients.}
receives. The well known result is that the only Nash equilibrium is for both vendors to locate at the median. This explains why sellers bunch together, but also why political candidates tend to have very similar platforms, converging on the opinion of the median voter.

However, there are many cases to which this observation does not apply. In the commercial setting, introducing price competition has been shown to cause competitors to differentiate in location~\cite{d1979hotelling,osborne1987equilibrium}.
Additional factors with a dispersing affect include transportation costs~\cite{osborne1987equilibrium}, congestion~\cite{ahlin2013product,feldotto2019hotelling,peters2018hotelling}, and queues~\cite{kohlberg1983equilibrium,peters2015waiting}.
Nevertheless, those considerations do not apply to the political setting, and explaining how a polarized political space may emerge~\cite{garimella2018political} remains a limitation of Hotelling's model.
Our motivating question in this paper concerns identifying and understanding some of the factors of the Hotelling game that drive competitors to disperse rather then cluster together. Our results provide a possible explanation of why in some settings it would pay off for political candidates or firms to diverge from their competition.

The model we study is motivated by the following insightful observation, pointed out by several other authors~\cite{feldman2016variations,shen2017hotelling,ben2017shapley}. One of the key assumptions at the basis of the Hotelling model is that clients will always go to the closest vendor, no matter how far he is. This assumption might be problematic in some settings. In the political context, for instance, the assumption means that voters may be willing to compromise their beliefs to an unlimited extent. In reality, this is not necessarily valid; it is possible that if no candidate presents sufficiently close opinions, the voter may simply abstain from voting.
\par
To address this issue, we adopt a modified variant of the Hotelling game, introduced and studied in~\cite{feldman2016variations,shen2017hotelling,ben2017shapley}, in which clients (voters, in the political context) have a limited {\em tolerance interval}, and a client will choose only players (candidates, in the political context) that fall within her tolerance interval.
In our model the interval boundaries are chosen randomly, as each client has a different tolerance threshold (reflecting, e.g.,
different degrees of openness to other political views).

It is important to note that our model deviates from the previous models in two central ways.
First, in our game, the player that the client chooses from among the eligible players (falling within her tolerance interval) is not arbitrary but rather the closest one (breaking ties uniformly at random).
This expresses the intuition that while a voter may be open minded and willing to vote to a candidate with a vastly divergent standpoint, she would still rather vote to a candidate that closely agrees with her own opinions provided one exists. Similarly, the proverbial sunbather would prefer to visit a closer vendor, even if she is willing to travel a longer distance when necessary. In this sense, our model maintains Hotelling's original intuition while capturing the realization that clients would not choose players that are too distant.

The second difference between our model and previous ones has to do with symmetry. Recently, a growing concern for the political discourse in western democracies is the phenomenon of {\em echo chambers}~\cite{garimella2018political}, namely, social media settings such as discussion groups and forums, in which one is exposed exclusively to opinions that agree with, and enhance, her own\footnote{There are several reasons this phenomenon is increasingly prevalent online. First, exposure to content is curated by algorithms according to each user's personal preferences. Second, on social media, users are more likely to share with their network content that agrees with their own opinion. Third, it has become increasingly easier to join private discussion groups that consist of like-minded individuals.}. This phenomenon tends to ``shorten'' the tolerance intervals of individual voters. But more importantly, we note that the echo chamber effect is very likely to act in a {\em one-sided} manner, making a voter more receptive to views on one side of the political spectrum than the other. Hence in certain settings, it is unreasonable to assume that a client has the same tolerance bounds on both sides.

To take such settings into account, we consider two variants of the game: \emph{symmetric games}, where clients have the same range of tolerance to their left and right, which expresses the willingness of a client to go a certain distance, with no preference of direction, and \emph{asymmetric games}, where the left and right ranges of each client are determined independently of each other, which captures settings where the scope of views each client is exposed to may be biased due to media bias, one-sided echo chambers, or tendencies in her local environment.

It may be natural to expect our results to depend heavily on the distribution according to which client tolerances are chosen. Surprisingly, it turns out that most of our general findings apply to a wide class of distributions.

\subsection{Contributions}
In our model, the left and right tolerance ranges  of each client are randomly distributed according to a given density $f$. Hence a game $G(n,f)$ is determined by the number of players $n$ and the distribution function $f$.
We consider two variants:
\emph{symmetric ranges} and \emph{asymmetric ranges.}
In a symmetric game the left and right ranges of each client are equal,
whereas in an asymmetric game the left and right ranges of each client
are independent and identically distributed random variables.

We start by characterizing the Nash equilibria of the 2-player game (Theorem~\ref{thm:NE-2players}). For $n\geq3$, we identify a specific class of strategy profiles, referred to as {\em canonical profiles}, where the distance between every pair of neighboring players is constant, and the distance from the leftmost player to 0 (a.k.a. the \emph{left hinterland}) is the same as the distance from the rightmost player to 1 (the \emph{right hinterland}).

We then show that canonical profiles are the only ones that may yield Nash equilibria in our game, namely, if there is an equilibrium then it must be canonical (Theorem~\ref{thm:NE-form}). Moreover, the canonical profile, when it exists, is uniquely defined for every $n$ and $f$. Hence, given a specific game $G(n,f)$, our problem is reduced to considering whether the canonical profile is a Nash equilibrium for given values of $n$ and $f$.

In the symmetric setting we give simple conditions for the canonical profile to be a Nash equilibrium, and demonstrate their application for several distributions. In the asymmetric setting, the conditions for equilibria are more complex, but we show that for some distributions, existence of a Nash equilibrium in the symmetric setting implies its existence in the asymmetric setting (Theorem~\ref{thm:sym-to-asym}). Finally, we show that even though Theorem~\ref{thm:sym-to-asym} does not apply for the exponential distribution, it is still possible to derive a full characterization of its Nash equilibria. Specifically, for the exponential distribution of parameter $\r$ in the asymmetric setting, we show that a Nash equilibrium exists for the $n$-player game if and only if $\r\geq \r_{\min}(n)$, for some \emph{threshold function}  $\r_{\min}(n)$ (Theorem~\ref{thm:NE-lower-bound}). Additionally, we show a way to efficiently approximate the values of $\r_{\min}(n)$ to any precision.

\subsection{Related Work}

Hotelling's model and its many variants have been studied extensively.
Downs \cite{downs1957economic} extended the Hotelling model to ideological positioning in a bipartisan democracy. It is remarkable to note that even in Downs' original work it was stipulated that extremists would rather abstain than vote to center parties, but no mathematical framework was provided for this property of the model.
Our work formalizes Downs' original intuition.
Eaton and Lipsey~\cite{eaton1975principle} extended Hotelling's analysis to any number of players and different location spaces. Our model is a direct extension of their $n$-player game on the line segment. d'Aspremont et al.~\cite{d1979hotelling} criticized Hotelling's findings and showed that when players compete on price as well as location, they tend to create distance from one another, otherwise price competition would drop their profit to zero. Our results show a differentiation in location in the $n$-player Hotelling game without introducing price competition. A large portion of the Hotelling game literature is dedicated to models with price competition. We, however, exclusively consider pure location competition models since they apply more directly to certain settings, such as the political one. Eiselt, Laporte and Thisse~\cite{eiselt1993competitive} provide an extensive comparison of the different models classified by the following characteristics: the number of players, the location space (e.g., circle, plane, network), the pricing policy, the behavior of players, and the behavior of clients. (For more recent surveys see Eiselt et al.~\cite{eiselt2011equilibria} and Brenner~\cite{brenner2010location}.)
Osborne et al.~\cite{osborne1993candidate} showed that in many variants of the political setting, no Nash equilibrium exists for more than $2$ players. In our model a Nash equilibrium exists for any number of players.

Randomness in client behavior was introduced by De Palma et al.~\cite{de1987existence}. Their model assumes client behavior has an unpredictable component due to unquantifiable factors of personal taste, and thus clients have a small probability of ``skipping'' the closest player and buying from another. In their model, all players would locate at the center in equilibrium, reasserting Hotelling's conclusion. In our model, clients exhibit randomness in their choice of players as well, but in equilibrium players create a fixed distance from their neighbors.

Several recent works~\cite{feldman2016variations,shen2017hotelling,ben2017shapley} studied the Hotelling model with limited attraction.
Feldman et al.~\cite{feldman2016variations} introduced the Hotelling model with limited attraction, where, similarly to our model, clients are unwilling to travel beyond a certain distance. They considered a simplified variant of the model where each player has an attraction interval of width $w$ for some fixed $w$. Their model admits an equilibrium for any number of players. Moreover, for most values of $w$, there exist infinitely many equilibria.
%
%
(In contrast, our model admits at most a single Nash equilibrium with a distinct structure.)
Shen and Wang~\cite{shen2017hotelling} extend the model of~\cite{feldman2016variations} to general distributions of clients. Ben-Porat and Tennenholz~\cite{ben2017shapley} consider random ranges of tolerance, and show that their game behaves like a cooperative game, since player payoffs are equal to their Shapley values in a coalition game. Their analysis relies on the fact that their game is a potential game, which does not hold for our model.
As explained above, our model diverges from these studies in other ways. In particular, in our model, clients are not allowed to ``skip'' over players, and must choose the closest player within their tolerance interval, whereas the previous studies assume clients are indifferent between players within their range of tolerance. Also, our model introduces the notion of asymmetric ranges of tolerance, which has not been considered before.

\section{Model}

Consider a setting in which clients are uniformly distributed along the interval $[0,1]$. A client is represented as a point $v\in[0,1]$, denoting her preference along the interval $[0,1]$. Clients are non-strategic.
\par
The strategic interaction in our model occurs between a finite set $N=\{1,\ldots,n\}$ of players. The set of strategies for a player is to choose a point in the interval $[0,1]$. Let $s_i\in [0,1]$ denote the strategy of player $i$, $1\leq i \leq n$. A strategy profile is given by a vector of player locations $\mathbf{s}=(s_1,\ldots,s_n)$. Let $\Sm i$ denote the profile of actions of all the players different from $i$. Slightly abusing notation, we denote by $(s'_i, \Sm i)$ the profile obtained from a profile $\s$ by replacing its $i$th coordinate $s_i$ with $s'_i$. We assume without loss of generality that $0 \leq s_1\leq \cdots \leq s_n \leq 1$. For the sake of notational convenience, we denote $s_0=0$ and $s_{n+1}=1$.

Each client $v$ has left and right ranges of tolerance denoted $\rangeR_v$ and $\rangeL_v$ respectively. The \emph{tolerance interval} of client $v$ is defined as $I_v=[v-\rangeL_v,v+\rangeR_v]$. The client $v$ supports the closest player within its tolerance interval. If there exists more than one closest player, then $v$ chooses one of the closest players uniformly at random.
Formally, $X(\s)=\{s_i \mid 1\leq i \leq n \}$ is the set of locations occupied by a player under $\s$. For every client $v$ the set of occupied locations inside $v$'s tolerance interval is denoted $T_v(\s)=X(\s)\cap I_v$.
Let $A_v(\s)=\arg\min_{x\in T_v(\s)}|x-v|$ be the location $v$ is attracted to. This set contains at most two locations, one to each side of $v$, but it is convenient to break ties by selecting the location on the left\footnote{There are at most $n-2$ points which are at equal distances from the nearest player on the right and on the left, and given that there is a continuum of clients in total, modifying $A_v(\s)$ in those points does not affect player utilities.}, i.e., if $A_v(\s)=\{s_i,s_j\}$ such that $s_i<s_j$, we modify $A_v(\s)$ to be $\{s_i\}$.
For every player $i$ and location $x$, the attraction of $v$ to location $x\in X(\s)$ is given by
$$ \omega_{v,x}(\s) ~=~
\Pr\left[\mbox{$v$ is attracted to a player in location $x$} \right]
~=~ \Pr\left[\{x\}=A_v(\s)\right]~.$$


We consider $\rangeR_v$, $\rangeL_v$ to be non-negative random variables drawn from the same distribution $\mathcal{D}$ independently for all clients $v$. We consider two variants of the game: \emph{symmetric} and \emph{asymmetric}. In the symmetric variant, $\rangeR_v=\rangeL_v$, for every client $v$. In the asymmetric variant, $\rangeR_v$ and $\rangeL_v$ are independent identically distributed random variables, for every client $v$. Throughout the paper, we denote by $f:[0,1] \to [0,1]$ the probability density function of $\mathcal{D}$, and the cumulative distribution function is denoted as $F:[0,1] \to [0,1]$. That is, $F(t)=\Prb[\rangeR_v\leq t]=\Prb[\rangeR_t\leq t]$. Additionally, for the analysis it is convenient to define $\bar{F}(t)=1-F(t)=\Prb[\rangeR_t\geq t]$.

Given a strategy profile $\mathbf{s}=(s_1,\ldots,s_n)$, two players $i,j\in N$ are said to be \emph{colocated} if $s_i=s_j$. For $i\in N$, the set of $i$'s colocated players is defined as $ \Gamma_i= \{j \in N \mid s_j=s_i\}$, and the size of this set is $\gamma_i= |\Gamma_i|$. A player that is not colocated with other players is called \emph{isolated}. Two players are called \emph{neighbors} if no player is located strictly between them. A \emph{left (resp., right) peripheral player} is a player that has no players to its left (resp., right). The players divide the line into \emph{regions} of two types: \emph{internal regions}, which are regions between two neighbors, and two \emph{hinterlands}, which include the region between 0 and the left peripheral player, and the region between the right peripheral player and 1. (See Figure~\ref{fig:NE-shape}, where the two hinterlands are marked by $a$.)

For $i\in N$, player $i$'s \emph{left} and \emph{right neighbors} are $L(s_i)=\max_{x\in X(\s)}\{x < s_i\}$ and $R(s_j)=\min_{x\in X(\s)}\{x > s_i\}$, respectively. Namely, these are the closest occupied player locations on either side of player $i$. We define $L(s_i)=0$ when $i$ does not have a left neighbor and $R(s_i)=1$ when $i$ does not have a right neighbor. Define the \emph{total utility at the location} of player $i$ as
$$ U_i(\s)=\TutilL_i(\s)+\TutilR_i(\s)=\int_{L(s_i)}^{R(s_i)}\omega_{v,s_i}(\s)dv~,$$
where $\TutilL_i(\s)$ and $\TutilR_i(\s)$ are the \emph{left} and \emph{right}
total utilities at player $i$'s location,
$$
    \TutilL_i(\s)=\int_{L(s_i)}^{s_i}\omega_{v,s_i}(\s)dv~~~~~~~~~~~~\mbox{and}~~~~~~~~~~~~~
    \TutilR_i(\s)=\int_{s_i}^{R(s_i)}\omega_{v,s_i}(\s)dv~.
$$
The total utility of player $i$ represents the total amount of clients attracted to location $s_i$ (either to player $i$ itself or to some colocated player $j\in \Gamma_i$).
\par
 The \emph{utility}, \emph{left utility} and \emph{right utility} of player $i$ are defined to be
$$
    u_i(\s)=\utilL_i(\s)+\utilR_i(\s)=\frac{U_i(\s)}{\gamma_i}~, ~~~~~~~~~~ \utilL_i(\s)=\frac{\TutilL_i(\s)}{\gamma_i}~~~~~~~\mbox{and}~~~~~~~
    \utilR_i(\s)=\frac{\TutilR_i(\s)}{\gamma_i}~.
$$

To summarize, our game is fully defined by the number of players $n$, the probability density function of client tolerances $f$, and whether the setting symmetric or asymmetric. Let $G^S=G^S(n,f)$ be the game under the symmetric setting, and let $G^A=G^A(n,f)$ be the game under the asymmetric setting. When making a claim that applies to both the symmetric and asymmetric setting we will use the notation $G=G(n,f)$.



Given a profile $\s$, $s'_i\in[0,1]$ is an \emph{improving move} for player $i$ if $u_i(s'_i,\Sm i) > u_i(\s)$. $s^*_i\in [0,1]$ is a \emph{best response} for player $i$ if $u_i(s^*_i,\Sm i) \geq u_i(s'_i,\Sm i)$ for every $s'_i\in [0,1]$.
\par
A profile $\s^*$ is a \emph{Nash equilibrium} if no player has an improving
move, i.e., for every $i\in N$ and every $s_i\in [0,1]$,
$u_i(\s^*) \geq u_i(s_i,\Sm i^*)$.

\section{Canonical Profiles as the only Possible Equilibria}
\label{sec:canonical-profile}

In this section we characterize a specific class of strategy
profiles, referred to as {\em canonical profiles}, and show that these
profiles are the only ones that may yield Nash equilibria in our game
(namely, if there is an equilibrium then it must be canonical).
We then show that each game $G(n,f)$ admits a unique canonical profile, $\Sc$,
if one exists.
This significantly simplifies later analysis and explains why the game
presents similar behavior for every number of players $n\geq 3$.
We conclude this section with a set of necessary and sufficient conditions
for a given canonical profile to be a Nash equilibrium.
Consequently, for every subclass of the game considered in the following
sections, it suffices to consider these conditions to either find
the entire set of Nash equilibria of a game $G(n,f)$ provided one exists,
or prove that the game admits no Nash equilibrium.

\subsection{Calculating Utilities}
Note that the utilities in our game are locally defined, i.e., the utility of player $i$ is independent of the location of players outside the interval $[L(s_i),R(s_i)]$. This is due to fact that a player $i$ may only attract clients from within $i$'s adjacent regions. Moreover, the attraction $\omega_{v,s_i}$ of a client $v\in[L(s_i),R(s_i)]$ to the location of player $i$ depends only on the distance $|v-s_i|$, the length of the region $v$ is inside, and whether it is a hinterland or an internal region.

It follows that the game $G(n,f)$ is uniquely determined by the following two functions:
\begin{align}
    \H(x)&=\int_0^x\Prb\big[\rangeR_t\geq t\big]dt \\
    \M(x)&=\int_0^\frac{x}{2}\Prb\big[\rangeL_t\geq t\big]dt + \int_\frac{x}{2}^x\Prb\big[\rangeL_t\geq t \; \wedge \; \rangeR_t<x-t\big]dt
\end{align}

Intuitively, $\H(x)$ (respectively, $\M(x)$) denotes the expected amount of support an isolated player gains from a hinterland (resp., an internal region)
of length $x$.
Note that in the symmetric setting, we have that $\rangeL_v=\rangeR_v$ for every client $v\in[0,1]$ and therefore, for every $t\in[x/2,x]$,
$$\Prb[\rangeL_t\geq t \, \wedge \, \rangeR_t<x-t]=0~.$$
However, in the asymmetric setting, $\rangeL_v$ and $\rangeR_v$ are independent random variables and thus
$$\Prb[\rangeL_t\geq t \, \wedge \, \rangeR_t<x-t]=\Prb[\rangeL_t\geq t]\cdot\Prb[\rangeR_t< x-t]~.$$
Recalling that $F(t)=\Prb[\rangeR_t\leq t]$ the next observation follows.
\begin{observation}\label{obs:H-and-M}
    For a symmetric game $G^S(n,f)$, the functions $\H$ and $\M$ are
    \begin{align}
        \H(x)&=\int_0^x(1-F(t))dt \label{eq:H-symm}\\
        \M(x)&=\int_0^\frac{x}{2}(1-F(t))dt \label{eq:M-symm}
    \end{align}
    For an asymmetric game $G^S(n,f)$, the functions $\H$ and $\M$ are
    \begin{align}
        \H(x)&=\int_0^x(1-F(t))dt \label{eq:H-asymm}\\
        \M(x)&=\int_0^\frac{x}{2}(1-F(t))dt+\int_\frac{x}{2}^x(1-F(t))F(x-t)dt
\label{eq:M-asymm}
    \end{align}
\end{observation}
Recalling that both $\rangeL_v$ and $\rangeR_v$ are drawn from the same distribution for all $v\in[0,1]$, the next claim follows by substituting variables in the integration.

\begin{lemma}\label{clm:H-and-M}
  For any game $G=G(n,f)$, profile $\s$
  s.t. $0\leq s_1\leq s_2 \leq \cdots \leq s_n \leq 1$,
  and $i\in N$,
        $$\TutilL_i(\s)=\left\{\begin{array}{ll}
                            \H(s_i), & \mbox{i is left peripheral;} \\
                            \M(s_i-L(s_i)), & \mbox{otherwise.}
                          \end{array} \right. $$
    and
        $$\TutilR_i(\s)=\left\{\begin{array}{ll}
                            \H(1-s_i), & \mbox{i is right peripheral;} \\
                            \M(R(s_i)-s_i), & \mbox{otherwise.}
                          \end{array} \right. $$
\end{lemma}

Note that since $\rangeL_t$ and $\rangeR_t$ are identically distributed for all $t\in[0,1]$, the functions $\H(x)$ and $\M(x)$ do not depend on whether the player is incident to the left or to the right of the region in question.

As an illustrative example, consider the profile $\mathbf{s}=(0.2, 0.5, 0.6)$ in a three player game $G(3,f)$. Then
$u_1(\mathbf{s})=\H(0.2)+\M(0.3)$,
$u_2(\mathbf{s})=\M(0.3)+\M(0.1)$ and
$u_3(\mathbf{s})=\M(0.1)+\H(0.4)$.

It is possible to define any game $G(n,f)$ by simply determining $\H(x)$ and $\M(x)$. In fact, these functions may be used to define many other variants of the Hotelling model not considered within the scope of this paper. Throughout this section, we will not use the explicit formulas for $\H(x)$ and $\M(x)$ and our results do not depend on these formulas. Instead, we derive our results based solely on the assumption that for the game under consideration, $\H(x)$ and $\M(x)$ satisfy the following properties:
\begin{description}
\item[(HM1)]\label{item:hm1}
Both functions
$\H$ and $\M$ are twice differentiable, concave and monotonically increasing,
That is,
for $x\in[0,1)$, $\H'(x)>0$, $\M'(x)>0$, $\H''(x)<0$ and $\M''(x)<0$.
\item[(HM2)]\label{item:hm2} For $x\in [0,1]$,  $\H(x)\geq \M(x)$.
\item[(HM3)]\label{item:hm3} $\H(0)=\M(0)=0~.$
\end{description}
Therefore, our results in this section are general and apply to any game $G(n,f)$ where (HM1), (HM2) and (HM3) are satisfied.

\subsection{Optimizing utilities locally}
We next extablish the optimal (maximum-utility) location of each player $i$ when fixing the locations of the other players and assuming $i$ can only move between its neighbors, but not ``jump'' over a neighbor.
Consider a peripheral player, and suppose its neighbor is at distance $0\le x\le 1$ from the endpoint.
For $0\le t\le x$, denote by $\theta_x(t)$ the utility of a peripheral player when its hinterland is of length $t$, and by $\mu_x(t)$ the utility of an internal player $i$ with $s_i-L(s_i)=t$ and $R(s_i)-L(s_i)=x$. By Lemma~\ref{clm:H-and-M},
$$ \theta_x(t) ~=~ \H(t)+\M(x-t)  ~~~~~~\mbox{ and }~~~~~~
\mu_x(t) ~=~ \M(t)+\M(x-t) $$
for $x\in[0,1]$ and $t\in [0,x]$.

\noindent {\em Remark.}
To keep $\theta$ and $\mu$ continuous in the interval $[0,t]$, we disregard the fact that for $x=t$ and $x=0$ the player is colocated with one of its neighbors, and assume all its payoff comes from the same interval of length $t$.
As shown in Claim~\ref{clm:nocolocated}, this assumption does not affect the analysis of the Nash equilibria of the game.


\begin{lemma}\label{lem:opt-local}
Let $G$ be game satisfying (HM1), (HM2), (HM3). For $x\in[0,1]$,
\begin{enumerate}[label=(\alph*)]
        \item \label{itm:lem-local-1} $\theta_x$ and $\mu_x$ are strictly concave functions of $t$.
        \item \label{itm:lem-local-2} $t=x/2$ is the unique maximum of $\mu_x$ in $[0,x]$.
        \item \label{itm:lem-local-3} If $\H'(x)>\M'(0)$, then $\theta_x$ is strictly increasing in $[0,x]$.
        \item \label{itm:lem-local-4} If $\H'(x)\leq\M'(0)$, then $t^*$ is the unique maximum of $\theta_x$ in $[0,x]$, where $t^*$ is the unique solution of the equation $\H'(t^*)=\M'(x-t^*)$~.
    \end{enumerate}
\end{lemma}
\begin{proof}
\noindent    \textbf{\ref{itm:lem-local-1}}
is immediate from the definitions of $\theta_x$ and $\mu_x$ and assumption (HM1).

\noindent    \textbf{\ref{itm:lem-local-2}}
    Since $\mu_x$ is concave, it has a unique maximum in every closed interval. Setting $t=x/2$ we obtain
    $$ \mu_x'\left(\frac{x}{2}\right)=
    \M''\left(\frac{x}{2}\right)-\M''\left(x-\frac{x}{2}\right)=0~, $$
    and thus $\mu_x$ is uniquely maximized at $t=x/2$.

\noindent    \textbf{\ref{itm:lem-local-3}}
    Since $\theta_x$ is concave, $\theta'_x$ is monotone decreasing. Therefore, for all $t\in[0,x]$,
    $$ \theta'_x(t)\geq \theta'_x(x)= \H'(x)+\M'(0)>0~.$$

\noindent    \textbf{\ref{itm:lem-local-4}}
    Note that for both (HM2) and (HM3) to hold we must have $\H'(0)\geq\M'(0)$, otherwise $\H(\varepsilon)<\M(\varepsilon)$ for sufficiently small $\varepsilon>0$. It follows by the monotonicity of $\M'$ that
        $\H'(0)\geq\M'(0)\geq \M'(x)~.$
    Hence, $\theta'_x(0)\geq0$. By assumption, $\theta'_x(x)\leq0$. So by the intermediate value theorem there must exist $t^*\in [0,x]$ such that $\theta'_x(t^*)=0$, i.e.,
    $$ \H'(t^*)=\M'(x-t^*)~, $$
    and since $\theta_x$ is concave, it follows that $t^*$ uniquely maximizes $\theta_x$ in the interval $[0,x]$.
\end{proof}

Let $\rho(x)$ denote the unique maximum of $\theta_x$ in the interval $[0,x]$. By Lemma~\ref{lem:opt-local}, the function $\rho:[0,1]\to[0,1]$ is well defined, and is given by
$$ \rho(x) = \left\{\begin{array}{ll}
                      x, & \mbox{if $\H'(x)>\M'(0)$;}\\
                      t^*, & \mbox{if $\H'(x)\leq\M'(0)$,}
                    \end{array}
                \right.$$
where $t^*$ is the unique solution of $\H'(t^*)=\M'(x-t^*)$.
We next show several properties of $\rho(x)$ that will be used in the proofs of our main results.

\begin{lemma}\label{clm:rho-properties}
Let $x<y$, for $x,y\in[0,1]$. Then
\begin{enumerate}[label=(\alph*)]
        \item \label{itm:clm-rho-1} If $\H'(x)\leq\M'(0)$, then $\H'(\rho(x))\leq\M'(0)$.
        \item \label{itm:clm-rho-2} $\rho(x)\leq\rho(y)$.
        \item \label{itm:clm-rho-3} $\theta_x(\rho(x))\leq\theta_{y}(\rho(y))$.
    \end{enumerate}
\end{lemma}
\begin{proof}
\noindent    \textbf{\ref{itm:clm-rho-1}}
    Since $\H'(x)\leq\M'(0)$, by Lemma~\ref{lem:opt-local}, it holds that $\H'(\rho(x))=\M'(x-\rho(x))$. Furthermore, $\M'$ is monotone decreasing, so the claim follows.

\noindent    \textbf{\ref{itm:clm-rho-2}}
    If $\H'(y)>\M'(0)$, then by Lemma~\ref{lem:opt-local}, $\rho(y)=y$. But $\rho(x)\leq x<y=\rho(y)$, so the claim holds.

    Next, suppose $\H'(y)\leq\M'(0)$. Assume towards contradiction that $\rho(x)>\rho(y)$. Since $x<y$, it follows that $x-\rho(x)<y-\rho(y)$. Therefore, since $\H'$ and $\M'$ are monotonically decreasing, we have that
    \begin{equation}\label{eq:rho-clm}
        \H'(\rho(x))<\H'(\rho(y))~~~~~~~\mbox{and}~~~~~~~\M'(x-\rho(x))>\M'(y-\rho(y))~.
    \end{equation}
    Moreover, $\H'(y)\leq\M'(0)$, so by definition of $\rho$ we have that
    $$ \H'(\rho(y))=\M'(y-\rho(y))~. $$
    Plugging this into Eq.~\eqref{eq:rho-clm}, we obtain
    $$ \H'(\rho(x))<\M'(x-\rho(x))~, $$
    which is a contradiction to the definition of $\rho$ (in both of its cases).

    \noindent    \textbf{\ref{itm:clm-rho-3}}
    Since the function $\theta_{y}$ is maximized at $\rho(y)$, we have that
    $$ \theta_{y}(\rho(y)) =\H(\rho(y))+\M(y-\rho(y))\geq \H(\rho(x))+\M(y-\rho(x))=\theta_{y}(\rho(x))~.$$
    But $\M$ is monotone increasing and $y>x$, so it follows that
$$ \theta_{y}(\rho(y))\geq \H(\rho(x))+\M(x-\rho(x))=\theta_{x}(\rho(x))~.$$
\end{proof}

\subsection{Nash Equilibria}
Let us now characterize the stable profiles that lead to a Nash equilibrium for a given game $ G(n,f)$.

The pair $\langle a,b \rangle$, $a,b\in[0,1]$, is called a \emph{canonical pair} if $a$ and $b$ satisfy the following equations:
\begin{align}
\label{eq:optimal-H}
&  \H'(a) = \M'(b)
\\
\label{eq:total-regions}
&  2a+(n-1)b = 1
\end{align}
Note that the canonical pair is unique for $n$ and $\r$ (if it exists at all).
A canonical pair induces a profile $\Sc$ for the game $G(n,f)$, such that
$$ s^{n,f}_i= a+(i-1)b$$
for every $i \in N$ (see Figure~\ref{fig:NE-shape}). We refer to this profile as a {\em canonical profile}.
\begin{figure}[ht]
\includegraphics[width=\textwidth]{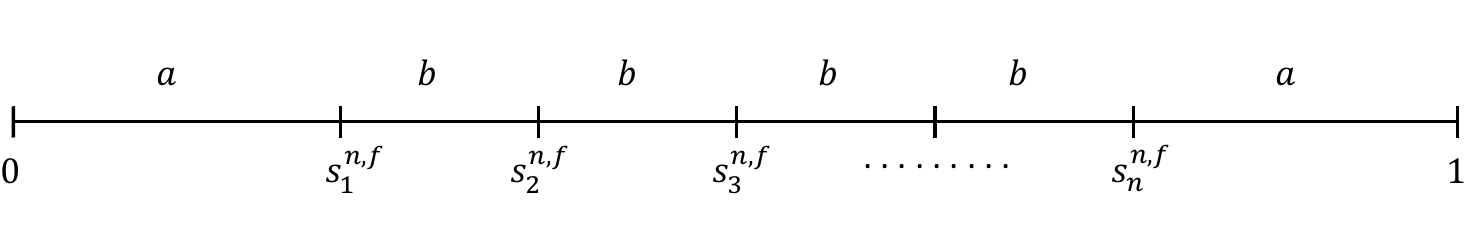}
\caption{a canonical profile.}
\label{fig:NE-shape}
\end{figure}
\begin{lemma}\label{lem:unique-canonic}
A game $G$ satisfying (HM1), (HM2), (HM3) has a canonical pair if and only if $\H'(1/2)\leq\M'(0)$. Moreover, if such a pair exists then it is unique.
\end{lemma}
\begin{proof}
    Let $\H'(1/2)>\M'(0)$, and assume towards contradiction that a canonical pair $\langle a,b \rangle$ exists. Note that $a\leq1/2$ and $b\geq0$ due to Eq.~\eqref{eq:total-regions}. Hence, since $\H'$ and $\M'$ are monotone decreasing, it follows that $\H'(a)\geq\H'(1/2)>\M'(0)\geq \M'(b)$, in contradiction to Eq.~\eqref{eq:optimal-H}. This shows that no canonical pair exists if $\H'(1/2)>\M'(0)$.

    Let $\H'(1/2)\leq\M'(0)$. We show that a canonical pair $\langle a,b \rangle$ exists. Note that for both (HM2) and (HM3) to hold we must have $\H'(0)\geq\M'(0)$, since otherwise $\H(\varepsilon)<\M(\varepsilon)$ for sufficiently small $\varepsilon>0$. Hence, by the monotonicity of $\M'$,
    \begin{equation}\label{eq:canonical-pair1}
        \H'(0)\geq\M'(0)\geq \M'\left(\frac{1}{n-1}\right)~.
    \end{equation}
    Consider the function $ g(x)= \H'(x)-\M'(\frac{1-2x}{n-1})$. By the assumption that $\H'(1/2)\leq\M'(0)$, we have that $g(1/2)\leq 0$, and by Eq.~\eqref{eq:canonical-pair1}, it holds that $g(0)\geq0$.
    Hence, by the intermediate value theorem there exists $x\in[0,1]$ such that $g(x)=0$. Thus, by definition, $\langle x,\frac{1-2x}{n-1} \rangle$ is a canonical pair. This shows that a canonical pair exists. It remains to show uniqueness.

    Assume towards contradiction that $\langle a,b \rangle$ and $\langle a',b' \rangle$ are both canonical pairs, such that $a'<a$. By Eq.~\eqref{eq:optimal-H},
        \begin{equation}\label{eq:canonical-pair2}
          \H'(a)=\M'(b) ~~~~~~~\mbox{and}~~~~~~~ \H'(a')=\M'(b')~.
        \end{equation}
    Since $a'<a$, by Eq.~\eqref{eq:total-regions}, it follows that $b'>b$. $\H'$ and $\M'$ are strictly decreasing functions, so
        $$\H'(a)<\H'(a') ~~~~~~~\mbox{and}~~~~~~~ \M'(b)>\M'(b')~,$$
    in contradiction to Eq.~\eqref{eq:canonical-pair2}. Therefore $a'=a$, which in turn ensures that also $b'=b$. This shows that the canonical pair is unique, and concludes the proof of the claim.
\end{proof}

\begin{theorem}\label{thm:NE-2players}
    Let $G$ be a game satisfying (HM1),(HM2) and (HM3), and let $n=2$. The game $G$ has a unique Nash equilibrium, which is given by
$$\s^*=\left\{\begin{array}{ll}
\;\;\Sc\;, & \mbox{if ~$\H'(1/2)\leq\M'(0)$;} \\
  \left(\frac12,\frac12\right),    & \mbox{otherwise.}
\end{array}\right. $$
\end{theorem}
\begin{proof}
    Let $(s_1,s_2)$ be a Nash equilibrium of the game. First note that $s_1\leq 1/2$, otherwise player 2 can improve its payoff by relocating to $1-s_2$ (due to the monotonicity $\H$ and $\M$). Symmetrically, we also have $s_2\geq 1/2$.

    Consider the case where $\H'(1/2)>\M'(0)$, and suppose that $s_2 > 1/2$. We have that $s_1\leq 1/2$, so it follows from the monotonicity of $\H'$ and $\M'$ that $\H'(s_1)\geq\H'(1/2)>\M'(0)\geq \M'(s_2-s)$. Therefore, by definition, $\theta_{s_2}(s_1)>0$. It follows by Lemma~\ref{lem:opt-local} that $\rho(s_2)$, the optimal location within the interval $[0,s_2]$, is not at $s_1$, contradiction. So suppose $s_2 = 1/2$ and thus $ \H'(s_2)>\M'(0) $. It follows, by Lemma~\ref{lem:opt-local}, that the optimal location for player 1 is at $s_1=1/2$.

    Now consider the case where $\H'(1/2)\leq\M'(0)$. It follows that $\H'(s_2) \leq \M'(0)$ and $ \H'(1-s_1) \leq \M'(0) $, and thus by Lemma~\ref{lem:opt-local}, we get
	\[
		\H'(s_1) = \M'(s_2-s_1)
	       \qquad  \hbox{ and }  \qquad
		\H'(1-s_2) = \M'(s_2-s_1)~,
	\]
    which yields $\H'(s_1) = \H'(1-s_2)$, and therefore $s_1=1-s_2$, due to the monotonicity of $\H'$. It follows that $(x_1,x_2)$ is the canonical profile $\Sc$. We have shown that each player cannot improve locally, it remains to show each player cannot improve by moving to the other hinterland. Consider player 1 relocating to segment $[s_2,1]$. By Claim~\ref{clm:rho-properties}\ref{itm:clm-rho-3}, since $s_2>1/2$, we have that $\theta_{1-s_2}(\rho(1-s_2))\leq\theta_{s_2}(\rho(s_2))$, and thus the optimal location in $[s_2,1]$ is not an improving move for player 1. This shows that player 1 has no improving move. Symmetrically, player 2 has no improving move, so this is a Nash equilibrium. Uniqueness follows from the uniqueness of the canonical pair (Lemma~\ref{lem:unique-canonic}).
\end{proof}

\noindent
For $n\ge 3$ players, we show that the only possible Nash equilibrium is the canonical profile.
Towards proving Theorem~\ref{thm:NE-form}, we state the following claim.

\begin{lemma}\label{clm:nocolocated}
Let $G$ be a game satisfying (HM1),(HM2) and (HM3), and let $n\geq 3$. If $\mathbf{s}$ is a Nash equilibrium then no two players are colocated in $\mathbf{s}$.
\end{lemma}
\begin{proof}
    Assume towards contradiction that there exist players $i,i+1, \ldots, i+k$ such that $s_{i-1}<s_i=s_{i+1}=\cdots=s_{i+k}<s_{i+k+1}$ for $k\geq1$. Suppose at first that these are internal players, i.e., $s_1<s_i<s_n$. Hence, by Lemma~\ref{clm:H-and-M} the utility of $i$ is given by
    $$ u_i(\mathbf{s})=\frac{\M(s_i-s_{i-1})+\M(s_{i+k+1}-s_i)}{k+1}~. $$
    Assume without loss of generality that $\M(s_i-s_{i-1})\geq \M(s_{i+k+1}-s_i)$.
    Note that
    $$\max(\M(s_i-s_{i-1}),\M(s_{i+k+1}-s_i))\geq \frac{\M(s_i-s_{i-1})+\M(s_{i+k+1}-s_i)}{2}~.$$
    Therefore, if $k>1$ or $\M(s_i-s_{i-1})\neq \M(s_{i+k+1}-s_i)$, then $s_i-\varepsilon$ is an improving move for player $i$, for sufficiently small $\varepsilon>0$. We may thus assume that $k=1$ and $\M(s_i-s_{i-1}) = \M(s_{i+2}-s_i)$ and thus $u_i(\mathbf{s})=\M(s_i-s_{i-1})$. But $\M$ is a concave function so $\M(s_i-s_{i-1})\leq 2\M((s_i-s_{i-1})/2)$. Therefore $(s_i+s_{i-1})/2$ is an improving move for player $i$, contradicting the assumption. We have shown that no Nash equilibrium $\s$ has colocated internal players. We next consider peripheral players.

    Suppose now that $i=1$, that is, $s_1=s_2=\cdots=s_{k+1}<s_{k+2}$ for $k<n-1$ (i.e., not all players are colocated). Then by Lemma~\ref{clm:H-and-M} the utility of player $k+1$ is given by
    $$ u_{k+1}(\mathbf{s})=\frac{\H(s_{k+1})+\M(s_{k+2}-s_{k+1})}{k+1}~. $$
    As in the previous case, if $k>1$ or $\H(s_{k+1})\neq \M(s_{k+2}-s_1)$, then player $1$ has an improving move. Hence, we assume $k=1$ and $\H(s_2)=\M(s_{3}-s_2)$. Therefore, $u_2(\mathbf{s})=\M(s_3-s_{2})$, and as above it follows from the concavity of $\M$ that $(s_3+s_2)/2$ is a local improving move for player 2.

    The only case left to consider is when all players are collocated at the same location, i.e., $s_1=s_2=\cdots=s_n$. Then by Lemma~\ref{clm:H-and-M} the utility of player 1 is given by
    $$ u_1(\mathbf{s})=\frac{\H(s_1)+\H(1-s_1)}{n}~. $$
    Since $n>2$, as in both previous cases player 1 can improve by moving
    to $s_1+\varepsilon$ for small enough $\varepsilon>0$. The claim follows.

\end{proof}

\begin{theorem}\label{thm:NE-form}
Let $G$ be a game satisfying (HM1),(HM2) and (HM3), and let $n\geq 3$. If the game $G$ admits a Nash equilibrium, then it is unique and equal to the canonical profile $\Sc$.
\end{theorem}

\begin{proof}
By Lemma~\ref{clm:nocolocated}, no two players are colocated in $\mathbf{s}$. Our proof relies on local optimization of the location of each player.
    Assume towards contradiction that $\H'(1/2)>\M'(0)$. By Lemma~\ref{lem:opt-local}, it follows that $u_1$, the utility of player 1, is an increasing function in the interval $[0,s_2]$ with respect to $s_1$, hence $(s_1+s_2)/2$ is an improving move for player 1, in contradiction to $\s$ being an equilibrium. Thus, it must hold that $\H'(1/2)
    \leq\M'(0)$~.

    Note that $s_1\leq 1/2$, otherwise player 2 can improve its payoff by relocating to $1-s_2$ (due to the monotonicity $\H$ and $\M$). Symmetrically, we also have $s_2\geq 1/2$.
    Thus, by Lemma~\ref{lem:opt-local}, we have that
    \begin{equation}\label{eq:ne-form1}
        \H'(s_1)=\M'(s_2-s_1) ~~~~~\mbox{and}~~~~~ \H'(1-s_n)=\M'(s_n-s_{n-1})~.
    \end{equation}
    Additionally, by Lemma~\ref{lem:opt-local}, we have that the optimal location of each internal player $i\in \{2,\dots,n-1\}$ is at $s_i=(s_{i-1}+s_{i+1})/2$. Consequently, there exists a constant $b$ such that $s_2-s_1=s_3-s_2=\cdots=s_n-s_{n-1}=b$. Pluggin $s_2-s_1=s_n-s_{n-1}$ into Eq.~\eqref{eq:ne-form1} we obtain
    $$ \H'(s_1) = \H'(1-s_n)~, $$
    and $\H'$ is monotone decreasing since $\H$ is concave, so we get $s_1=1-s_n=a$, proving that $\mathbf{s}$ is a canonical profile. Uniqueness follows from Lemma~\ref{lem:unique-canonic}.
\end{proof}

\begin{lemma}\label{lem:ne-condition}
Let $G$ be a game satisfying (HM1),(HM2) and (HM3), and let $n\geq 3$.
Let $\Sc$ be the canonical profile of $G$, with a corresponding canonical pair
$\langle a,b \rangle$. Then $\Sc$ is a Nash equilibrium if and only if
\begin{align}
\H(a)+\M(b) &~\geq~~ 2 \M\left(\frac{b}{2}\right), \label{eq:ne-condition-peripheral}\\
\H(\rho(a))+\M(a-\rho(a)) &~\leq~~ 2\M(b)~. \label{eq:ne-condition-internal}
\end{align}
\end{lemma}
\begin{proof}
    By Lemma~\ref{lem:opt-local} we have that in each region there is one point that maximizes the payoff. Moreover, by Lemma~\ref{clm:nocolocated} the payoff gained from colocating with another player $j$ can always be exceeded by that of an isolated location within one of $j$'s adjacent regions. Therefore, the canonical profile $\Sc$ is a Nash equilibrium if and only if relocating to the optimal location of any region is not an improving move for any player $i$.

    Note that $\Sc$ is constructed such that no player $i$ can improve by moving to one of its adjacent regions, so it remains to consider whether $i$ can improve by moving to a non-adjacent region. It is easy to see that no internal player has an improving move in an internal region and that no peripheral player has an improving move in the other hinterland (due to the monotonicity of $\H$ and $\M$).

    Let us consider the possibility that the peripheral player 1 has an improving move in an internal region. By Lemma~\ref{clm:H-and-M}, we have that $u_1(\Sc)=\H(a)+\M(b)$, and by Lemma~\ref{lem:opt-local} the optimal utility player 1 would gain by relocating to an internal segment is $2\M(b/2)$. It follows that if Eq.~\eqref{eq:ne-condition-peripheral} holds, then player 1 has no improving move.

    We next consider whether an internal player $i\in\{2,\ldots,n-1\}$ has an improving move in the hinterland. By Lemma~\ref{clm:H-and-M}, we have that $u_i(\Sc)=2\M(b)$, and by Lemma~\ref{lem:opt-local} the optimal utility player 1 would gain by relocating to a hinterland is $\H(\rho(a))+\M(a-\rho(a))$. It follows that if Eq.~\eqref{eq:ne-condition-internal} holds, then player $i$ has no improving move. This shows that Eq.~\eqref{eq:ne-condition-peripheral} and~\eqref{eq:ne-condition-internal} are necessary and sufficient conditions for the canonical profile $\Sc$ to be a Nash equilibrium.
\end{proof}

\section{Symmetric Range Distributions}
\label{sec:symm-distr}

In this section we consider the symmetric game $G^S(n,f)$, where the range of each client $v$ satisfies $\rangeL_v=\rangeR_v$. We first show that the game satisfies assumptions (HM1), (HM2) and (HM3), which allows us to use all the results of Section~\ref{sec:canonical-profile}.

\subsection{General Properties}
Recall that $\bar{F}(t)=1-F(t)$, where $F$ is the cumulative distribution function of $f$.

\begin{lemma}~\label{lem:symm-H-and-M-concave}
Let $G^S(n,f)$ be a game such that $f$ is continuously differentiable and has full support (i.e., $f(x)>0$ for all $x\in[0,1]$), then $G^S(n,f)$ satisfies assumptions (HM1), (HM2) and (HM3).
\end{lemma}
\begin{proof}
    (HM2) and (HM3) follow immediately from Observation~\ref{obs:H-and-M}, so it remains to show (HM1).
    Note that by Observation~\ref{obs:H-and-M}, we have that $\M(x)=\H(x/2)$ for every $x\in[0,1]$,  so it suffices to show that $\H$ is increasing and concave. By the fundamental theorem of calculus we obtain
        $$ \H'(x)= 1-F(x)~, $$
    which is strictly positive due to the fact that $f(x)>0$ for all $x\in[0,1)$ and thus not all of the mass of the distribution is contained in $[0,x]$. This shows that $\H$ is increasing.

    We next show $\H$ is concave. By definition of $F$, the second derivative is
         $$ \H''(x)= -f(x)~, $$
    which is strictly negative by the assumption on $f$. It follows that $\H$ is concave. This proves the lemma.
\end{proof}

Due to Lemma~\ref{lem:symm-H-and-M-concave}, we may apply Theorem~\ref{thm:NE-2players} and Theorem~\ref{thm:NE-form} to the game $G^S(n,f)$ and we thus obtain the following two corollaries.

\begin{corollary}\label{cor:symm-NE-2player}
    Let $n=2$, and let $f$ be continuously differentiable and have full support (i.e., $f(x)>0$ for all $x\in[0,1]$). Then the game $G^S(2,f)$ has a unique Nash equilibrium, which is given by
$$ \s^*=\left\{\begin{array}{ll}
\;\;\Sc\;,		         & \mbox{if ~$\bar{F}(1/2)\leq1/2$;} \\
\left(\frac12,\frac12\right),    & \mbox{otherwise,}
\end{array}\right. $$
where $\Sc$ is the canonical profile, with a corresponding canonical pair $\langle a,b \rangle$, where $a$ and $b$ are given implicitly by the equation
    $\bar{F}(a)=\bar{F}(b/2)/2~.$
\end{corollary}
\begin{corollary}\label{cor:symm-canonical}
Let $n\geq 3$. For every game $G^S(n,f)$ where $f$ is continuously differentiable and has full support (i.e., $f(x)>0$ for all $x\in[0,1]$), if $G^S$ admits a Nash equilibrium then it is unique (up to renaming the players) and equal to the canonical profile $\Sc$, with a corresponding canonical pair $\langle a,b \rangle$, where $a$ and $b$ are given implicitly by the equation
$\bar{F}(a)=\bar{F}(b/2)/2~.$
\end{corollary}

\begin{lemma}\label{clm:symm-rho-bound}
Under the symmetric game $G^S(n,f)$, the function $\rho:[0,1]\to[0,1]$ satisfies that $\rho(x)>x/3$, for all $x\in[0,1]$.
\end{lemma}

\begin{proof}
    Recall that $\rho(x)=\arg\max_{y\in[0,x]}\theta_x(y)$. Furthermore, by Lemma~\ref{lem:opt-local}, $\theta_x$ is concave, so it follows that if $\theta'_x(y)\geq0$, (i.e., $\theta_x$ is non-decreasing at the point $y$) then $\rho(x)\geq y$. Thus, it suffices to show that $\theta_x(x/3)\geq 0 $ for every  $x\in[0,1]$. Plugging Eq.~\eqref{eq:H-symm} and~\eqref{eq:M-symm} into the definition of $\theta_x$ we obtain
    $$\theta_x(y) = \int_0^y(1-F(t))dt+\int_0^{\frac{x-y}{2}}(1-F(t))dt~.$$
    By the fundamental theorem of calculus, the derivative with respect to $y$ is given by
    $$\theta'_x(y) = 1-F(y)-\frac{1}{2}\left(1-F\left(\frac{x-y}{2}\right)\right)~.$$
    Plugging in $y=x/3$ yields
    $$\theta'_x\left(\frac{x}{3}\right) = 1-F\left(\frac{x}{3}\right)-\frac{1}{2}\left(1-F\left(\frac{x}{3}\right)\right) =
    \frac{1}{2}\left(1-F\left(\frac{x}{3}\right)\right)\geq 0~. $$
    This proves the claim.
\end{proof}

\begin{lemma}\label{lem:symm-equilib}
    The game $G^S(n,f)$ as in Lemma~\ref{lem:symm-H-and-M-concave}, for $n\geq3$, admits a Nash equilibrium if and only if
$$ \H(\rho(a))+\M(a-\rho(a)) \leq 2\M(b)~, $$
where $\rho(a)$ is defined implicitly by the equation
$\bar{F}(\rho(a))=\bar{F}((a-\rho(a))/2)/2$.
\end{lemma}
\begin{proof}
    By Lemma~\ref{lem:symm-H-and-M-concave}, $G^S(n,f)$ satisfies assumptions (HM1), (HM2) and (HM3). Therefore, by Lemma~\ref{lem:ne-condition}, $G^S(n,f)$ admits a Nash equilibrium if and only if Eq.~\eqref{eq:ne-condition-peripheral} and \eqref{eq:ne-condition-internal} are satisfied. Recall that by the definition of the canonical profile, $a$ is the optimal location for player 1 within the interval $[0,a+b]$, i.e., $\rho(a+b)=a$. Thus, by Lemma~\ref{clm:symm-rho-bound}, it follows that $a\geq(a+b)/3$, which yields $a\geq b/2$. by assumption (HM2) and the monotonicity of $\H$ and $\M$, we obtain
    $$ \H(a)+\M(b)\geq \M(a)+\M(b) \geq 2\M\left(\frac{b}{2}\right) $$
    Thus, Eq.~\eqref{eq:ne-condition-peripheral} always holds in the symmetric setting. Hence, the game admits a Nash equilibrium if and only if Eq.~\eqref{eq:ne-condition-internal} is satisfied. The lemma follows.
\end{proof}

We conclude the discussion of symmetric games with a number of example
distributions and their equilibria states.

\subsection{Example 1: The Uniform Distribution}
This distribution, in which the range boundary parameter $\rangeL_v$ ($=\rangeR_v$) is drawn uniformly at random from $[0,1]$,
was considered by Ben-Porat and Tennenholtz~\cite{ben2017shapley} in a setting where clients are allowed to skip over players. Here we show that, if ``skipping'' is not allowed, as in our model, there is no Nash equilibrium for $n\geq 3$. For $n=2$, the only Nash equilibrium is $( 1/2,1/2 )$, where both players are colocated at the center. The probability density function and corresponding cumulative density function are defined as
$$
\mbox{$f(x)= \left\{\begin{array}{l l}
                    1,  & x \in [0,1];  \\
                    0,   & \mbox{otherwise,}
                \end{array}\right.$}
~~~~~~~~~~
\mbox{$F(x)= \left\{\begin{array}{l l}
                    x,  & x \in [0,1];  \\
                    1,  & x \geq 1;     \\
                    0,   & \mbox{otherwise.}
                 \end{array}\right.$}
$$

\begin{proposition}\label{prp:symm-uniform}
    For the game $G^S(n,f)$, where $f$ is the uniform distribution, there exists a Nash equilibrium if and only if $n=2$, and it is equal to the strategy profile $(1/2,1/2)$.
\end{proposition}
\begin{proof}
Plugging $f$ into Corollary~\ref{cor:symm-canonical}, we get that the only Nash equilibrium of $G^S(n,f)$ is the canonical profile represented by the canonical pair $\langle a,b \rangle$ where $a$ and $b$ satisfy
$ 1-a=(1-b/2)/2  $
By definition, the canonical pair additionally satisfies $2a+(n-1)b=1$. So the only solution is $a=1/2$ and $b=0$. Hence, $G^S(n,f)$ has a Nash equilibrium if and only if $n=2$.
\end{proof}


\subsection{Example 2: Linear Distributions}
Here we consider any distribution whose density is linear and whose mass is entirely contained in $[0,1]$. Specifically, assume $\int_0^1(rx+q)dx=1$, or rather $q=1-r/2$. For $f(x)$ to be non-negative in $[0,1]$ we also need $-2\leq r \leq 2$. Then take
$$
\mbox{$f(x)= \left\{\begin{array}{l l}
                    rx+q,  & x \in [0,1];  \\
                    0,   & \mbox{otherwise,}
                \end{array}\right.$}
~~~~~~~~~~
\mbox{$ F(x)= \left\{\begin{array}{l l}
                    \frac{r}{2}x^2+qx,  & x \in [0,1];  \\
                    1,  & x \geq 1;     \\
                    0,   & \mbox{otherwise.}
  \end{array}\right.$}
$$
To make the analysis cleaner let us pick the two extreme examples of the parameters $(r,q)$, namely, $(-2,2)$ and $(2,0)$.

\begin{proposition}\label{prp:symm-linear}
The game $G^S(n,f)$, where $f$ is the linear distribution with coefficients either $(r_1,q_1)=(-2,2)$ or $(r_2,q_2)=(2,0)$, has a Nash equilibrium if and only if $n=2$.
%
\begin{enumerate}[label=(\alph*)]
\item \label{prp:linear-1}
For $r_1,q_1$, the only Nash equilibrium is (1/2,1/2).
\item \label{prp:linear-2}
For $r_2,q_2$, the only (canonical) Nash equilibrium is given by the canonical pair
\begin{equation*}
a=\frac{2\sqrt{2}+1}{2\sqrt{2}+2} ~~~~~~~~\mbox{and}~~~~~~~~
b=\frac{\sqrt{2}-1}{1+1/\sqrt{2}}~.
\end{equation*}
\end{enumerate}
\end{proposition}
\begin{proof}
    For $r_2,q_2$ we get that $\bar{F}(1/2)>1/2$ and thus by Corollaries~\ref{cor:symm-NE-2player} and~\ref{cor:symm-canonical} the only Nash equilibrium is $(1/2,1/2)$ when $n=2$.

    For $r_1,q_1$ we get the equation $(1-a)^2=\frac{1}{2}(1-b/2)^2$, which yields
    \begin{equation*}
     a=\frac{2(\sqrt{2}-1)n-2\sqrt{2}+3}{2\sqrt{2}(n-1)+2} ~~~~~~~~\mbox{and}~~~~~~~~ b=\frac{\sqrt{2}-1}{n-1+1/\sqrt{2}}~.
    \end{equation*}

    If $n=2$, a Nash equilibrium always exists, so this is a Nash equilibrium. For $n\geq 3$, by Lemma~\ref{lem:symm-equilib} this is a Nash equilibrium if for every $x$ in $[0,a]$ we have
    \begin{equation}
    \label{eq:abx3}
    \left(1-\frac{a-x}{2}\right)^3+(1-x)^3 \geq 2 \left(1-\frac{b}{2}\right)^3~.
    \end{equation}
    But $a\geq 2b$ so we can choose $x$ such that $x\geq b/2$ and $(a-x)/2\geq b/2$, making the left hand side of Ineq.~\eqref{eq:abx3} strictly smaller than the right hand side, contradiction. Hence there is no Nash equilibrium.
\end{proof}

\noindent {\em Remark.}
Intuitively, in the uniform distribution the players are forced to converge towards the center. In comparison, in the linear distribution corresponding to $(r_2,q_2)$ it is likelier for clients to have a large range, which means more clients inside the hinterland will be covered by the peripheral player, so it will be beneficial for it to move closer to its neighbor and have fewer clients contested by another player, despite having a greater average distance to potential clients.

\subsection{Example 3: Pareto Distributions}
The distribution $\pareto(\alpha,\xi)$ for parameters $\alpha>0$ and $\xi>0$ has density function and cumulative distribution function
$$
\mbox{$f(x)= \left\{\begin{array}{l l}
                    0,   & x < \xi~;         \\
                    \frac{\alpha\xi^\alpha}{x^{\alpha+1}},  & x \geq \xi~,
                \end{array}\right.$}
~~~~~~~~~~
\mbox{$F(x)= \left\{\begin{array}{l l}
                    0,   & x < \xi~;    \\
                    1-\left(\xi/x\right)^\alpha,  & x \geq \xi~.
               \end{array}\right.$}
$$

\begin{proposition}\label{prp:symm-pareto}
    For the game $G^S(n,f)$, where $f$ is the density of $\pareto(\alpha,\xi)$, the canonical pair is given by
    \begin{equation}
         a=\frac{2^{1/\alpha-1}}{n-1+2^{1/\alpha}} ~~~~~~~~\mbox{and}~~~~~~~~ b=\frac{1}{n-1+2^{1/\alpha}}~,
    \end{equation}
    and it is a Nash equilibrium if and only if $\alpha\geq z$, where $z$ is the unique solution of the equation $ 2^{1/z}(2+2^{1/z})^{z} = 8 $ such that $0<z<1$.
\end{proposition}
\begin{proof}
    Note that $f(x)=0$ for $x\in[0,\xi)$ so the conditions of Corollary~\ref{cor:symm-canonical} do not apply. However, for $\xi$ small enough the result still holds.\footnote{The Theorem does not hold in cases where $b<2\xi$. This is due to the fact that all clients have a range of at least $\xi$, and therefore the internal player gets the same utility at all points of the interval between its neighbors. But in such cases an internal player would have an improving move in the hinterlands, so the canonical profile is not an equilibrium anyway.} The equation $\bar{F}(a)=\bar{F}(b/2)/2$ yields
    $$ \left(\frac{\xi}{a}\right)^\alpha = \frac12 \left(\frac{2\xi}{b}\right)^\alpha~,$$
    which translates to $a=2^{1/\alpha-1}b$. By definition the canonical pair satisfies $2a+(n-1)b=1$. So we get
    \begin{equation}\label{eq:pareto-canonical}
         a=\frac{2^{1/\alpha-1}}{n-1+2^{1/\alpha}} ~~~~~~~~\mbox{and}~~~~~~~~ b=\frac{1}{n-1+2^{1/\alpha}}~.
    \end{equation}

    To apply Lemma~\ref{lem:symm-equilib} we first calculate
    $$ \int_0^x\bar{F}(t)dt = \left\{\begin{array}{l l}
                        \xi+\xi\ln\left(\frac{x}{\xi}\right),   & \alpha=1~;   \\
                        \xi+\frac{\xi}{\alpha-1}\left(1-\left(\frac{\xi}{x}\right)^{\alpha-1}\right),  & \alpha\neq1~.
                    \end{array}\right.
    $$
    Let $x^*\in[0,a]$ satisfy
    $$ \bar{F}(x^*)=\bar{F}((a-x^*)/2)/2.$$
    This solves to
    \begin{equation}\label{eq:pareto-x}
    x^*=\frac{2^{1/\alpha}a}{2+2^{1/\alpha}}  ~~~~~~~~\mbox{and}~~~~~~~~
    \frac{a-x^*}{2}=\frac{a}{2+2^{1/\alpha}}
    \end{equation}

    Consider $\alpha=1$ first. Note that in this case, $a=b$ and $x^*=a/2$, so the condition given in Lemma~\ref{lem:symm-equilib} translates to
        $$
            \int_{\frac{a}4}^{\frac{a}2}\bar{F}(t)dt \geq \int_{\frac{a}2}^{\frac{a}2}\bar{F}(t)dt~.
        $$
    This always holds since the left hand side is non-negative and the right hand side is zero. So we have that for $\alpha=1$ the canonical pair is $\langle \frac{1}{n+1},\frac{1}{n+1} \rangle$ and it is a Nash equilibrium.

    Let us now consider $\alpha\neq1$. By Lemma~\ref{lem:symm-equilib} the canonical profile given in \eqref{eq:pareto-canonical} is a Nash equilibrium if and only if
        $$
            \frac{\xi}{\alpha-1}\left(\left(\frac{2\xi}{a-x^*}\right)^{\alpha-1}
                                    -\left(\frac{2\xi}{b}\right)^{\alpha-1}\right)\geq
            \frac{\xi}{\alpha-1}\left(\left(\frac{2\xi}{b}\right)^{\alpha-1}
                                    -\left(\frac{\xi}{x^*}\right)^{\alpha-1}\right)~.
        $$
    Rearranging, we obtain
        $$ \frac{1}{\alpha-1}\left(\left(\frac{2}{a-x^*}\right)^{\alpha-1}+\left(\frac{1}{x^*}\right)^{\alpha-1}\right)\geq
                \frac{2}{\alpha-1}\left(\frac{2}{b}\right)^{\alpha-1}~.$$
    Plugging $b=2^{1-1/\alpha}a$ and Eq.~\eqref{eq:pareto-x} in the above we get
        $$  \frac{1}{\alpha-1}\left(\left(\frac{2+2^{1/\alpha}}{a}\right)^{\alpha-1}+
      \left(\frac{2+2^{1/\alpha}}{2^{1/\alpha}a}\right)^{\alpha-1}\right) \geq
        \frac{2}{\alpha-1}\left(\frac{2^{1/\alpha}}{a}\right)^{\alpha-1}~.$$
    Suppose $\alpha>1$. We get
        $$ (2+2^{1/\alpha})^{\alpha-1}\left(\frac1{2^{(\alpha-1)/\alpha}}+1\right) \geq 8\cdot 2^{-1/\alpha} $$
    or $2^{1/\alpha}(2+2^{1/\alpha})^{\alpha} \geq 8$.
    This inequality holds since equality holds for $\alpha=1$ and the left hand side is a monotonically increasing function of $\alpha$ for $\alpha>1$.

    Now consider $\alpha<1$. We get $ 2^{1/\alpha}(2+2^{1/\alpha})^{\alpha} \leq 8~. $
    There exists a constant $0<z<1$ such that $ 2^{1/z}(2+2^{1/z})^{z} = 8 $.
    The inequality holds for $\alpha\in [z,1]$.

    To summarize, we have that for $\alpha\geq z$ the game has a unique Nash equilibrium which is the canonical profile given in Eq.~\eqref{eq:pareto-canonical}.
\end{proof}

\subsection{Example 4: Exponential Distributions}
The exponential distribution with parameter $\lambda>0$ has density function
and cumulative density function
$$
\mbox{$f(x)= \left\{\begin{array}{l l}
                    0,   & x < 0~;         \\
                    \lambda e^{-\lambda x},  & x \geq 0~,
                \end{array}\right.$}
~~~~~~~~~~
\mbox{$ F(x)= \left\{\begin{array}{l l}
                    0,   & x < 0~;    \\
                    1-e^{-\lambda x},  & x \geq 0~.
  \end{array}\right.$}
$$

\begin{proposition}\label{prp:symm-exp}
    For the game $G^S(n,f)$, where $f$ is the density of the exponential distribution with parameter $\lambda>0$, the canonical pair is given by
    $$ a= \frac{1}{n}\left(\frac12+\frac{(n-1)\ln2}{\lambda}\right)  ~~~~~~~~\mbox{and}~~~~~~~~ b=\frac{1}{n}\left(1-\frac{2\ln2}{\lambda}\right) $$
    and it is a Nash equilibrium if and only if $\lambda\geq \ln4-n\ln(4\tau_1^6)$, where $\tau_1 = \frac{\sqrt[6]{2} + \sqrt{64 + \sqrt[3]{2}}}{8 \cdot 2^{5/6}}\approx 0.65$.
\end{proposition}
\begin{proof}
    By Corollary~\ref{cor:symm-canonical} we have that the only Nash equilibrium of the game $G^S(n,f)$ is the canonical pair $ \langle a,b \rangle $ satisfying
    $$ e^{-\lambda a} =  e^{-\lambda b/2} / 2~,$$
    and thus $a=b/2+\ln2/\lambda$. By definition, $a$ and $b$ satisfy $2a+(n-1)b=1$ and therefore we obtain
    $$ a= \frac{1}{n}\left(\frac12+\frac{(n-1)\ln2}{\lambda}\right)  ~~~~~~~~\mbox{and}~~~~~~~~ b=\frac{1}{n}\left(1-\frac{2\ln2}{\lambda}\right) $$

    By Lemma~\ref{lem:symm-equilib} the canonical profile is a Nash equilibrium if and only if
    \begin{equation}\label{eq:sym-exp-condition}
            \int_{\frac{a-x^*}2}^{\frac{b}2}e^{-\lambda t}dt \geq \int_{\frac{b}2}^{x^*}e^{-\lambda t}dt
    \end{equation}
    where $x^*\in[0,a]$ is the unique solution of
    $$ e^{-\lambda x^*} = e^{-\lambda(a-x^*)/2} / 2~. $$
    So we get
    $$ x^*=\frac{a}{3}+\frac{2\ln2}{3\lambda}  ~~~~~~~~\mbox{and}~~~~~~~~ \frac{a-x^*}{2}= \frac{a}{3}-\frac{\ln2}{3\lambda}~.$$
    Plugging the above along with $b/2=a-\ln2 ~/\lambda$ into Eq.~\eqref{eq:sym-exp-condition} we get
    $$
    \sqrt[3]{2}e^{-\lambda a/3}+e^{-2\lambda a/3} / {\sqrt[3]{4}} \geq 4 e^{-\lambda a}~.
    $$
    Let $\tau=e^{-\lambda a / 3}$. The above inequality translates into
    $$
    4\tau^3-\tau^2 / {\sqrt[3]{4}} - \sqrt[3]{2}\tau \leq 0~.
    $$
    Solving the inequality we get $0\leq \tau \leq \tau_1 $, where $\tau_1 = \frac{\sqrt[6]{2} + \sqrt{64 + \sqrt[3]{2}}}{8 \cdot 2^{5/6}}\approx 0.65$.
    Plugging $\tau=e^{-\lambda a/3}$ back in we obtain
    $$  \lambda a \geq -3\ln\tau_1~, $$
    and plugging in $ a= \frac{1}{n}\left(\frac12+\frac{(n-1)\ln2}{\lambda}\right)$ we obtain
    $$ \lambda \geq \ln4-n\ln(4\tau_1^6)\approx 1.39+1.24n $$
    as the necessary and sufficient condition for the existence of a Nash equilibrium under the exponential distribution.
\end{proof}

\section{Asymmetric Range Distributions}
\label{sec:asymm-distr}

In this section we consider the asymmetric game $G^A(n,f)$, where the range boundaries of each client $v$, $\rangeL_v$ and $\rangeR_v$, are drawn independently at random.
%
As in the previous section, we begin by showing that the game satisfies assumptions (HM1), (HM2) and (HM3), allowing us to use the results of Section~\ref{sec:canonical-profile}.

\subsection{General Properties}

\begin{lemma}~\label{lem:asymm-H-and-M-concave}
    Let $G^A(n,f)$ be a game such that $f$ is continuously differentiable and has full support (i.e., $f(x)>0$ for all $x\in[0,1]$), then $G^A(n,f)$ satisfies assumptions (HM1), (HM2) and (HM3).
\end{lemma}
\begin{proof}
    (HM2) and (HM3) follow immediately from Observation~\ref{obs:H-and-M}, so it remains to show (HM1).
    Note that by Observation~\ref{obs:H-and-M}, $\H$ is the same as in the symmetric game $G^S(n,f)$, and thus by Lemma~\ref{lem:symm-H-and-M-concave}, $\H$ is increasing and concave. It remains to show that $\M$ is monotone and concave. Using the Leibniz integral differentiation rule we obtain
        $$ \M'(x) = \frac{1}{2}\left(1-F\left(\frac{x}{2}\right)\right)^2+
                    \int_\frac{x}{2}^x(1-F(t))f(x-t)dt>0~.$$
    Hence, $\M$ is monotone increasing. Taking the second derivative with respect to $x$ yields
        $$ \M''(x) = -\left(1-F\left(\frac{x}{2}\right)\right)f\left(\frac{x}{2}\right)+(1-F(x))f(0)
        +\int_{\frac{x}2}^{x}(1-F(t))f'(x-t)dt$$
    and using integration by parts we get
    $$\M''(x) = -\left(1-F\left(\frac{x}{2}\right)\right)f\left(\frac{x}{2}\right)+(1-F(x))f(0)+(1-F(t))f(x)|_{\frac{x}2}^{x}-\int_{\frac{x}2}^{x}f(t)(f(x-t))dt$$
    $$ = -\int_{\frac{x}2}^{x}f(t)f(x-t)dt <0 $$
    Hence, the function $\M$ is concave, and thus the lemma holds.
\end{proof}

Due to Lemma~\ref{lem:asymm-H-and-M-concave}, we may apply Theorem~\ref{thm:NE-2players} and Theorem~\ref{thm:NE-form} to the game $G^A(n,f)$ and we thus obtain the following two corollaries.

\begin{corollary}\label{cor:asymm-NE-2player}
    Let $n=2$, and let $f$ be continuously differentiable and have full support (i.e., $f(x)>0$ for all $x\in[0,1]$). Then the game $G^A(2,f)$ has a unique Nash equilibrium, which is given by
    $$ \s^*=\left\{
				\begin{array}{ll}
					\;\;\Sc\;,		         & \mbox{if ~$\bar{F}(1/2)\geq1/2$;} \\
                    \left(\frac12,\frac12\right),    & \mbox{otherwise,}
				\end{array}
			\right. $$
    where $\Sc$ is the canonical profile, with a corresponding canonical pair $\langle a,b \rangle$, where $a$ and $b$ are given implicitly by the equation
    $$
        \bar{F}(a)=\frac12\left(\bar{F}\left(\frac{b}2\right)\right)^2+\int_{\frac{b}2}^b\bar{F}(t)f(b-t)dt
    $$
\end{corollary}

\begin{corollary}\label{cor:asymm-canonical}
    Let $n\geq 3$. For every game $G^A(n,f)$ where $f$ has full support (i.e., $f(x)>0$ for all $x\in[0,1]$), if $G^A$ admits a Nash equilibrium it is unique (up to renaming the players) and equal to the canonical profile $\langle a,b \rangle$ where $a$ and $b$ are given implicitly by the equation
$$\bar{F}(a) ~=~ \frac12\left(\bar{F}\left(\frac{b}2\right)\right)^2
+\int_{\frac{b}2}^b\bar{F}(t)f(b-t)dt~.$$
\end{corollary}

\begin{lemma}\label{clm:asymm-rho-bound}
Under the asymmetric game $G^A(n,f)$, the function $\rho:[0,1]\to[0,1]$ satisfies that $\rho(x)>x/3$, for all $x\in[0,1]$.
\end{lemma}
\begin{proof}
    Recall that $\rho(x)=\arg\max_{y\in[0,x]}\theta_x(y)$. Furthermore, by Lemma~\ref{lem:opt-local}, $\theta_x$ is concave, so it follows that if $\theta'_x(y)\geq0$, (i.e., $\theta_x$ is non-decreasing at the point $y$) then $\rho(x)\geq y$. Thus, it suffices to show that $\theta_x(x/3)\geq 0 $ for every  $x\in[0,1]$. Plugging Eq.~\eqref{eq:H-asymm} and~\eqref{eq:M-asymm} into the definition of $\theta_x$ we obtain
    $$\theta_x(y) = \int_0^y(1-F(t))dt+\int_0^{\frac{x-y}{2}}(1-F(t))dt
                    +\int_{\frac{x-y}2}^{x-y}(1-F(t))F(x-y-t)dt~.$$
    By Leibniz's integral rule, the derivative with respect to $y$ is given by
    $$\theta'_x(y) = 1-F(y)-\frac{1}{2}\left(1-F\left(\frac{x-y}{2}\right)\right)^2 - \int_{\frac{x-y}2}^{x-y}(1-F(t))f(x-y-t)dt~.$$
    Plugging in $y=x/3$ yields
    $$\theta'_x\left(\frac{x}{3}\right) = 1-F\left(\frac{x}{3}\right)-\frac{1}{2}\left(1-F\left(\frac{x}{3}\right)\right)^2 - \int_{\frac{x}3}^{\frac{2x}3}(1-F(t))f\left(\frac{2x}{3}-t\right)dt~, $$
    using differentiation by parts and rearranging we obtain
    $$\theta'_x\left(\frac{x}{3}\right) = \frac{1}{2}\left(1-F\left(\frac{x}{3}\right)\right)^2 + \int_{\frac{x}3}^{\frac{2x}3}f(t)f\left(\frac{2x}{3}-t\right)dt \geq 0~, $$
    which proves the claim.
\end{proof}

\begin{lemma}\label{lem:asymm-equilib}
The game $G^A(n,f)$ admits a Nash equilibrium if and only if
\begin{equation}\label{eq:asymm-ne-condition}
\H(\rho(a))+\M(a-\rho(a)) \leq 2\M(b)~,
\end{equation}
where $\rho(x)$ is defined implicitly by the equation
$$\bar{F}(\rho(x))=\frac12\left(\bar{F}\left(\frac{a-\rho(x)}2\right)\right)^2
+\int_{\frac{a-\rho(x)}2}^{a-\rho(x)}\bar{F}(t)f(a-\rho(x)-t)dt~.$$
\end{lemma}
\begin{proof}
    By Lemma~\ref{lem:asymm-H-and-M-concave}, $G^A(n,f)$ satisfies assumptions (HM1), (HM2) and (HM3). Therefore, by Lemma~\ref{lem:ne-condition}, $G^A(n,f)$ admits a Nash equilibrium if and only if Eq.~\eqref{eq:ne-condition-peripheral} and \eqref{eq:ne-condition-internal} are satisfied. Recall that by the definition of the canonical profile, $a$ is the optimal location for player 1 within the interval $[0,a+b]$, i.e., $\rho(a+b)=a$. Thus, by Lemma~\ref{clm:asymm-rho-bound}, it follows that $a\geq(a+b)/3$, which yields $a\geq b/2$. Therefore, by assumption (HM2) and the monotonicity of $\H$ and $\M$, we obtain
    $$ \H(a)+\M(b)\geq \M(a)+\M(b) \geq 2\M\left(\frac{b}{2}\right) $$
    Thus, Eq.~\eqref{eq:ne-condition-peripheral} always holds in the asymmetric setting. Hence, the game admits a Nash equilibrium if and only if Eq.~\eqref{eq:ne-condition-internal} is satisfied. The lemma thus follows.
\end{proof}

\begin{theorem}\label{thm:sym-to-asym}
    If $G^S(n,f)$ admits a Nash equilibrium and
    $$ \frac{\partial}{\partial x}\left(\int_{\frac{x}2}^x\bar{F}(t)F(b-t)dt\right) \geq 0, $$
    then $G^A(n,f)$ admits a Nash equilibrium.
\end{theorem}
\begin{proof}
    If $n=2$, by Theorem~\ref{thm:NE-2players}, the theorem holds. Assume $n\geq 3$. To simplify notation, let $\s=\Sc$, be the canonical profile of $G^S(n,f)$. Let $\langle a,b \rangle$ the canonical pair of $G^S(n,f)$. Let $\widetilde{u}_i$ denote the utility of player $i$ in the game $G^A(n,f)$. Similarly, throughout this proof, we use the notation $\widetilde{\H}$, $\widetilde{\M}$ and $\widetilde{\rho}$
    when referring to $G^A(n,f)$, while the original notation refers to $G^S(n,f)$. Let $$\Delta(x)=\int_{\frac{x}2}^x\bar{F}(t)F(b-t)dt~.$$

    Note that by Observation~\ref{obs:H-and-M} we have that $\widetilde{\M}(x)=\M(x)+\Delta(x)$ and $\widetilde{\H}(x)=\H(x)$ for every $x\in[0,1]$, and therefore for every internal player $i\in\{2,\ldots,n-1\}$ we have that
        $$ u_i(\s)=\widetilde{u}_i(\s)+2\Delta(b) $$
    and for peripheral players
        $$ u_1(\s)=\widetilde{u}_1(\s)+\Delta(b) ~~~~~\mbox{and}~~~~ u_n(\s)=\widetilde{u}_n(\s)+\Delta(b)~.$$
    Let $x^*=\rho(a)$ denote the best response in the hinterland $(0,a)$ in the game $G^S(n,f)$, and let $\widetilde{x}^*=\widetilde{\rho}(a)$ denote the best response in the hinterland $(0,a)$ in the game $G^A(n,f)$. Since $\s$ is a Nash equilibrium in $G^S(n,f)$ we have that
        $$ u_i(\s)\geq u_i(x^*,\Sm{i}), $$
    which can be written as
    \begin{equation}\label{eq:sym-to-asym-1}
        2\M(b)\geq \H(x^*)+\M(a-x^*)
    \end{equation}

    Assume towards contradiction that $a-\widetilde{x}^*> b$. Due to Lemma~\ref{clm:asymm-rho-bound} we get $b\leq 2\widetilde{x}^*$, and it thus by monotonicity of $\M$ and since $\M(x)=\H(x/2)$ we get
    $$  u_i(\s) = 2\M(b) < \M(2\widetilde{x}^*)+\M(a-\widetilde{x}^*)=\H(\widetilde{x}^*)+\M(a-\widetilde{x}^*)=u_i(\widetilde{x}^*,\Sm{i})~,$$
    in contradiction to $\s$ being a Nash equilibrium in the symmetric game $G^S(n,f)$. Hence, $a-\widetilde{x}^*\leq b$. Moreover, by assumption, $\Delta(x)$ is increasing with respect to $x$, so we have that
    $$ \Delta(a-\widetilde{x}^*) \leq \Delta(b)$$
    Therefore,
    \begin{align*}
        \widetilde{u}_i(\widetilde{x}^*,\Sm{i}) = \widetilde{\H}(\widetilde{x}^*)+\widetilde{\M}(a-\widetilde{x}^*) &= \H(\widetilde{x}^*)+\M(a-\widetilde{x}^*)+\Delta(a-\widetilde{x}^*) \\
        &\leq  \H(x^*)+\M(a-x^*)+\Delta(a-\widetilde{x}^*)~,
    \end{align*}
    where the last inequality is due to the optimality of $x^*$ with respect to the function $\H(t)+\M(a-t)$. But $\Delta(a-\widetilde{x}^*) \leq \Delta(b)$, so plugging Equation~\eqref{eq:sym-to-asym-1} in the above we get
    \begin{equation}\label{eq:tilde-improve}
        \widetilde{\H}(\widetilde{x}^*)+\widetilde{\M}(a-\widetilde{x}^*) \leq 2\M(b)+\Delta(b) \leq 2\widetilde{\M}(b)=\widetilde{u}_i(\s)
    \end{equation}
    Therefore, there is no improving move for an internal player in $G^A(n,f)$ under $\s$.

    Let $\langle \widetilde{a},\widetilde{b} \rangle$
    \footnote{Note that by Lemma~\ref{lem:unique-canonic} $\H'(1/2)\leq\M'(0)$, and thus $\widetilde{\H}'(1/2)\leq \widetilde{\M}'(0)$, so the canonical pair $\langle \widetilde{a},\widetilde{b} \rangle$ exists.}
    be the canonical pair of $G^A(n,f)$, and let $\widetilde{\s}=\widetilde{\s}^{n,f}$ be the corresponding canonical profile. If we show that $\widetilde{b}\geq b$, then Equation~\eqref{eq:asymm-ne-condition} would hold, which would conclude the proof by Lemma~\ref{lem:asymm-equilib}. This is due to the fact that
    $$ 2\widetilde{\M}(\widetilde{b})\geq 2\widetilde{\M}(b) \geq \widetilde{\H}(\widetilde{x}^*)+\widetilde{\M}(a-\widetilde{x}^*)
                            \geq \widetilde{\H}(\widetilde{\rho}(\widetilde{a}))+\widetilde{\M}(\widetilde{a}-\widetilde{\rho}(\widetilde{a}))~, $$
    where the first inequality is due to the monotonicity of $\widetilde{\M}$, the second inequality is due to Equation~\eqref{eq:tilde-improve}, and the last inequality is due to Lemma~\ref{clm:rho-properties}\ref{itm:clm-rho-3}. Hence, it remains to show that $\widetilde{b}\geq b$.

    We will show that $\widetilde{a}\leq a$, which is equivalent. Consider the function $g(x)=\widetilde{\H}'(x)-\widetilde{\M}'((1-2x)/(n-1))$. By the definition of the canonical pair, we have that $g(\widetilde{a})=0$. Moreover, by concavity of $\widetilde{\H}$ and $\widetilde{\M}$ we have that $g$ is decreasing, and by Lemma~\ref{lem:unique-canonic} it has at most one zero. Hence, every $x$ such that $g(x)\leq0$ must satisfy $x \geq \widetilde{a}$. Furthermore,
    $$g(a)=\widetilde{\H}'(a)-\widetilde{\M}'((1-2a)/(n-1))=\H'(a)-\M'(b)-\Delta'(b)=-\Delta'(b)\leq0~, $$
    where the second and third equalities are due to $\langle a,b \rangle$ being a canonical pair, and the last inequality is due to the assumption. Thus, $\widetilde{a}\leq a$, which concludes the proof.
\end{proof}

Again we conclude with a couple of example distributions and their equilibria.



\subsection{Example 1: The Uniform Distribution}
This is the same as the uniform distribution for symmetric games, except that
{\em both} range boundary parameters $\rangeL_v$ and $\rangeL_v$ need to be drawn uniformly at random. We have the following.

\begin{proposition}\label{prp:asymm-uniform}
For the game $G^A(n,f)$, where $f$ is the uniform distribution, there exists
a Nash equilibrium if and only if $n=2$, and it is equal to the strategy
profile $(1/2,1/2)$.
\end{proposition}
\begin{proof}
    Using Corollary~\ref{cor:asymm-canonical}, we get
     $$ 1-a = \frac12\left(1-\frac{b}2\right)^2+\int_{\frac{b}{2}}^{b}(1-t)dt~, $$
    which translates to
     $$ 1-a =1/2-b^2/4~. $$
    By definition $ \langle a,b \rangle $ satisfies $2a+(n-1)b=1$. So, as in the symmetric setting, we get that the only solution is $a=\frac12$ and $b=0$. Hence, $G^A(n,f)$ has a Nash equilibrium if and only if $n=2$.
\end{proof}

\subsection{Example 2: The Exponential Distribution}
Finally, we consider the game $G^A(n,f)$ where $f$ is the density function of the exponential distribution with parameter $\r>0$. That is, the range of each client $v$ is asymmetric and exponentially distributed, i.e., $\rangeL_v,\rangeR_v\in \expD(\r)$. Slightly abusing notation, we refer to this game as $G^A(n,\r)$.
We dedicate special attention to this distribution for three main reasons.
First, the exponential distribution is commonly considered in geometric models, and has been shown to apply to many real life situations. Second, the game $G^A(n,\r)$ is mathematically equivalent to a \emph{fault-prone Hotelling game}, studied in our related paper \cite{AP18}. In this game, faults occur at random along the line, and clients cannot visit players separated from them by a random fault. Hence, our results on the exponential distribution can be applied directly to fully characterize the equilibria of another interesting variant of the Hotelling model. Finally, this example demonstrates that even though the condition of Theorem~\ref{thm:sym-to-asym} does not always apply to the exponential distribution, and the condition given for the existence of Nash equilibria in Lemma~\ref{lem:asymm-equilib} is somewhat hard to work with, it is nevertheless possible to fully analyze certain useful classes of client range distributions.

By Corollary~\ref{cor:asymm-NE-2player}, if $n=2$ then the game always admits a Nash equilibrium, which is the canonical profile if it exists, and $(1/2,1/2)$ otherwise.
For $n\geq 3$,
by Corollary~\ref{cor:asymm-canonical}, if the game admits a Nash equilibrium then it is equal to the canonical profile. Consequently, to fully characterize the equilibria of the game, the following theorem determines, for any given $n$ and $\r$, whether the canonical pair of $G^A(n,\r)$ is Nash equilibrium.
More precisely, the theorem characterizes a \emph{threshold function} $\r_{\min}(n)$ such that the game $G^A(n,\r)$ admits a Nash equilibrium if and only if $\r \geq \r_{\min}(n)$.

Moreover, the \emph{exact} formulation of the threshold function $\r_{\min}(n)$ depends on a global constant $\a_0$.
While there is no closed form formula for $\a_0$, it is implicitly defined as the unique solution of Eq.~\eqref{eq:alpha-beta6} and \eqref{eq:alpha-beta7} in the interval $[0,1]$ (see Figure~\ref{fig:param-a-b}), which is approximately $\a_0 \approx 0.58813$.

\begin{theorem}\label{thm:NE-lower-bound}
$G^A(n,\r)$ for $n\geq3$ admits a Nash equilibrium if and only if
$$\r ~\geq~ \r_{\min}(n) ~=~ (n+1)\a_0-2\ln\left(\frac{1+\a_0}2\right)~,$$
where $\alpha_0\in(0,1)$ is the unique constant given implicitly as the
solution to the
%
%
following equations:
\begin{align}
	e^{-\a}(1+\a) &=e^{-2\b}(1+\b)~,  \label{eq:alpha-beta6}\\
	e^{-\a}\left(1+\frac\a2\right) &=e^{-\b}\left(\frac34+\frac{\b}2\right)~. \label{eq:alpha-beta7}
\end{align}
Moreover, $\a_0 \approx 0.58813$, implying
that a Nash equilibrium exists if and only if
$$\r ~\geq~ \r_{\min}(n) ~\approx~ 0.58813n+1.04931~.$$
\end{theorem}

We start our analysis by deriving the functions $\H$ and $\M$ by plugging the cumulative distribution function of $\expD(\r)$, i.e., $F(t)=1-e^{-\r t}$ into Eq.~\eqref{eq:H-asymm} and~\eqref{eq:M-asymm}. This yields
\begin{align}
  \H(x)= &\frac{1}{\r}\left[1-e^{-\r x}\right] \label{eq:H-exp}\\
  \M(x)= &\frac{1}{\r}\left[1-e^{-\r x}\left(1+\frac{\r x}{2}\right)\right] \label{eq:M-exp}
\end{align}

We can characterize the canonical pair of the game $G^A(n,\r)$, provided it exists.
By Corollary~\ref{cor:asymm-canonical}, given an integer $n\geq 2$ and a real $\r > 0$, the canonical pair $\langle a,b \rangle$ exists if and only if $\r>2\ln2$ and is given by the following equations:
\begin{align}
\label{eq:optimal-H-exp}
&  e^{\lambda(b-a)} = \frac{1+\lambda b}2
\\
\label{eq:total-regions-exp}
&  2a+(n-1)b = 1
\end{align}


  By Lemma~\ref{lem:asymm-equilib}, in order to prove Theorem~\ref{thm:NE-lower-bound}, it suffices to show that making a non-local move, namely, relocating an internal player to $\rho(a)$, the optimal location within the hinterland, is not improving. Plugging Eq.~\eqref{eq:H-exp} and~\eqref{eq:M-exp} into the definition of $\rho(x)$, we obtain that in the game $G^A(n,\r)$, for every $x\in[\ln2/\r,1]$, $\rho(x)$ is given implicitly by
\begin{equation}\label{eq:exp-rho}
    e^{\lambda(x-2\rho(x))} = \frac{1+\lambda (x-\rho(x))}2
\end{equation}

\begin{lemma}\label{lem:region-difference}
    Let $x\in [\ln2/\r,1]$. Then,
        $$ \H(\rho(x))=\M(x-\rho(x))+\frac{1}{2\r}e^{-\r(x-\rho(x))} $$
\end{lemma}
\begin{proof}
    By Eq.~\eqref{eq:H-exp}, $$\H(\rho(x))=\frac{1}{\r}\left[1-e^{-\r \rho(x)}\right]$$
    Plugging in Eq.~\eqref{eq:exp-rho}, we get
    \begin{align*}
        \H(\rho(x))  &= \frac1\r\left[1-e^{-\r (x-\rho(x))}\left(\frac{1+\r(x-\rho(x))}2\right)\right] \\
                       &= \frac1\r\left[1-e^{-\r (s-\rho(x))}\left(1+\frac{\r(x-\rho(x))}2\right)\right]
                                                +\frac{1}{2\r}e^{-\r(x-\rho(x))}
    \end{align*}
    Plugging in Eq.~\eqref{eq:M-exp} we obtain the result.
\end{proof}

Before we continue, we define a reparamatrization that will greatly simplify the following analysis. Define
\begin{equation}
\label{eq:def-alpha}
\a=\r b   ~~~~\mbox{and}~~~~  c=1-a/b~.
\end{equation}


\begin{lemma}\label{clm:reparam}
The values $c$, $b$, $a$ and $\r$ can be expressed in terms of $\a$ and $n$ as follows.
\begin{description}
\item{(1)} $c=\ln\left(\frac{1+\a}{2}\right) / \a$; 
\item{(2)} $b=1 / (n+1-2c)$;                          
\item{(3)} $a=(1-c) / (n+1-2c)$;                        
\item{(4)} $\r=\a(n+1)-2\ln((1+\a)/2)$.  
    \end{description}
\end{lemma}
\begin{proof}
Substituting $c$ and $\a$ into Eq.~\eqref{eq:optimal-H-exp} we obtain $ e^{c\a} = (1+\a)/2$,
which yields the first part of the claim.
Plugging $a=(1-c)b$ into Eq.~\eqref{eq:total-regions-exp} yields the next two parts.
To obtain the last part we plug the terms we obtained for $b$ and $c$ into $\r=\a/b$.
\end{proof}

\begin{observation}\label{obs:monotone-lambda}
    The parameter $\r$ is monotone increasing as a function of $\a$ for all $\a>0$ and $n>2$. Therefore, $\a$ is a monotone increasing function of $\r$ as well.
\end{observation}

The observation follows from the fact the $\a$ derivative of $\r$ is strictly positive.

Keeping $n$ fixed, by Lemma~\ref{clm:reparam}, for each $\a>0$ we obtain $\r$, $a$ and $b$, such that $\langle a,b \rangle$ is the canonical pair of $G^A(n,\r)$. Moreover, by Observation~\ref{obs:monotone-lambda}, considering $\r$ as a function of $\a$ over the domain $\a\in(0,\infty)$, $\r$ obtains all the values $\r\in(2\ln2,\infty)$.

\begin{lemma}\label{lem:M_geq_H}
    If $b\geq a$, then $\Sc$ is a Nash equilibrium.
\end{lemma}
\begin{proof}
    By Lemma~\ref{lem:asymm-equilib}, it suffices to show that the canonical pair $a$ and $b$ satisfies Eq.~\eqref{eq:asymm-ne-condition} if $b \geq a$. Consider the right hand side of Eq.~\eqref{eq:asymm-ne-condition}. We have that
        $$ \H(\rho(a))+\M(a-\rho(a)) \leq \H(\rho(a))+\H(a-\rho(a))
            \leq 2\H\left(\frac{a}2\right)
            \leq 2\H\left(\frac{b}2\right)~.$$
    The first inequality holds since $\M(x)\leq\H(x)$ for all $x\geq0$,
    the second is due to the fact that $\H$ is concave and thus $\H(x)+\H(a-x)\leq 2\H(a/2)$ for every $x\in[0,a]$, and the last follows from the assumption and monotonicity.

    Therefore, Eq.~\eqref{eq:asymm-ne-condition} is satisfied if
        $$ \H\left(\frac{b}2\right) \leq \M(b)~. $$
    Plugging in Eq.~\eqref{eq:H-exp} and \eqref{eq:M-exp} we obtain
        $$\frac{1}{\r} \left[1-e^{-\r b/2}\right] \leq \frac{1}{\r} \left[1-e^{-\r b}\left(1+\frac{\r b}{2}\right)\right]~.$$
    Rearranging we obtain
        $$ e^{\r b/2} \geq 1+\frac{\r b}{2}~, $$
    which holds for all $\r$ and $b$, concluding the proof.
\end{proof}

\begin{observation}\label{obs:alpha_geq_one}
	$\a\geq 1$ if and only if $b \geq a$.
\end{observation}

\begin{proof}
By the definition of $c$ in Eq. \eqref{eq:def-alpha},
$M\geq H$ when $c \geq 0$.
By Part (1) of Lemma~\ref{clm:reparam}, $c$ is non-negative if and only if $\a\geq1$.
\end{proof}

\begin{lemma}\label{lem:alpha-beta}
Fix $n \geq 3$, let $\a>0$ be as in Lemma~\ref{clm:reparam}, and
    express $\r$ as in Part (4)  
    of Lemma~\ref{clm:reparam}.
Then $\Sc$ is a Nash equilibrium if and only if there exists $\b>0$ such that
        \begin{equation}\label{eq:alpha-beta1}
          e^{-\a}(1+\a)=e^{-2\b}(1+\b)~\phantom{.}
        \end{equation}
  and
        \begin{equation}\label{eq:alpha-beta2}
          e^{-\a}\left(1+\frac\a2\right) \leq e^{-\b}\left(\frac34+\frac\b2\right)~.
        \end{equation}
\end{lemma}

\begin{proof}
    Let $t=a-\rho(a)$.
    By Lemma~\ref{lem:region-difference},
    $$ \H(\rho(a))+\M(a-\rho(a)) = \H(a-t)+\M(t) = 2\M(t)+\frac1{2\r}e^{-\r t}~, $$
    so we may write Eq.~\eqref{eq:asymm-ne-condition} as
    $$ \M(b) \geq \M(t)+\frac1{4\r}e^{-\r t}~. $$
    Plugging in Eq.~\eqref{eq:M-exp} and rearranging we obtain
    $$ e^{-\r b}\left(1+\frac{\r b}2\right) \leq e^{-\r t}\left(\frac34+\frac{\r t}2\right)~. $$
    Setting $\a=\r b$ and $\b=\r t$ in the above, we obtain Eq.~\eqref{eq:alpha-beta2}. Note that $t$ is uniquely determined by $\a$, and therefore $\b$ may be derived from $\a$.

    By Eq.~\eqref{eq:exp-rho},
    $$  e^{-\r a} = e^{-2\r t} \left(\frac{1+\r t}2\right)~, $$
    and plugging in Eq.~\eqref{eq:optimal-H-exp} we get
    $$ e^{-\r b} \left(\frac{1+\r b}2\right)=e^{-2\r t} \left(\frac{1+\r t}2\right)~,$$
    which translates to Eq.~\eqref{eq:alpha-beta1}.

    Since each $\a>0$ defines a unique canonical profile and a unique $\b$, it follows that Eq.~\eqref{eq:alpha-beta1} and \eqref{eq:alpha-beta2} are equivalent to Eq.~\eqref{eq:asymm-ne-condition}. Therefore, by Lemma~\ref{lem:asymm-equilib}, Eq.~\eqref{eq:alpha-beta1} and \eqref{eq:alpha-beta2} are necessary and sufficient conditions for a Nash equilibrium.
\end{proof}

We are now ready to prove Theorem~\ref{thm:NE-lower-bound}.

\par\noindent {\em Proof of Theorem~\ref{thm:NE-lower-bound}.~}
By Observation~\ref{obs:alpha_geq_one} and Lemma~\ref{lem:M_geq_H}, if $\a\geq1$ then the canonical profile is a Nash equilibrium. It is left to consider $0<\a<1$. By Lemma~\ref{lem:alpha-beta}, Eq.~\eqref{eq:alpha-beta1} and \eqref{eq:alpha-beta2} are sufficient and necessary conditions for a Nash equilibrium. We rewrite these equations as follows.
 \begin{align}
	e^{-\a}(1+\a) &=e^{-2\b_1}(1+\b_1)~, \label{eq:alpha-beta3} \\
	e^{-\a}\left(1+\frac\a2\right) &=e^{-\b_2}\left(\frac34+\frac{\b_2}2\right)~, \label{eq:alpha-beta4} \\
	\b_1 &\leq \b_2~. \label{eq:alpha-beta5}
 \end{align}

Clearly, Eq.~\eqref{eq:alpha-beta1} and \eqref{eq:alpha-beta2} hold for $\a$ if and only if Eq.~\eqref{eq:alpha-beta3},\eqref{eq:alpha-beta4} and \eqref{eq:alpha-beta5} hold for that $\a$.
Figure~\ref{fig:param-a-b} shows Eq.~\eqref{eq:alpha-beta3} and \eqref{eq:alpha-beta4} as parametric curves.

\begin{figure}[ht]
    \begin{center}
        \includegraphics[height=6.5cm]{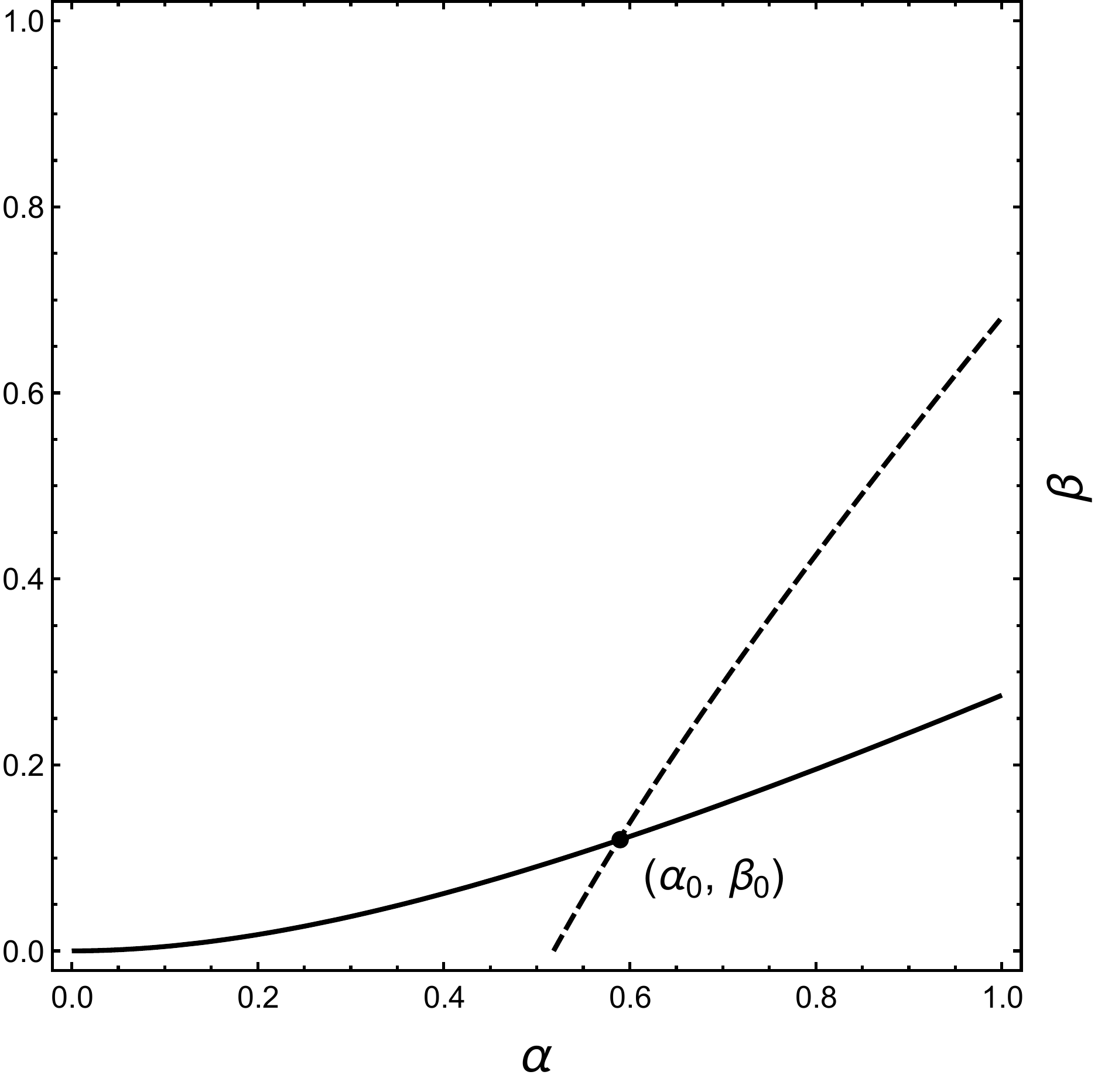}
    \end{center}
    \caption{The solid line corresponds to Eq.~\eqref{eq:alpha-beta3} and the dashed line corresponds to equation~\eqref{eq:alpha-beta4}. $(\a_0,\b_0)$ is the unique intersection of the two curves.}
    \label{fig:param-a-b}
\end{figure}

We next show that for $0<\a<1$ these two curves intersect at a single point $(\a_0,\b_0)$, as depicted in Figure~\ref{fig:param-a-b}. First note that by the implicit function theorem these two curves are continuous and differentiable for $0\leq\a\leq1$. Additionally, it is easy to check that if $\a=0$ then $\b_1<\b_2$ and that if $\a=1$ then $\b_1>\b_2$. Hence, the two curves intersect at least once for $0<\a<1$. To show they intersect \emph{exactly} once, we consider the $\a$ derivatives of $\b_1$ and $\b_2$, and show that, at each intersection point, $d\b_2/d\a>d\b_1/d\a$. Since $\b_1$ and $\b_2$ are continuous and continuously differentiable as functions of $\a$, this suffices to show that they intersect only once.

Accordingly, we now show that $d\b_2/d\a>d\b_1/d\a$ at each intersection point. By implicit differentiation (i.e., taking the $\a$ derivative of both sides of an implicit function) we obtain
$$ -e^{-\a}\,\a  = -2 e^{-2\b_1}(1+2\b_1) \, \frac{d \b_1}{d \a} $$
and
$$ -e^{-\a}\left(\frac{1+\a}{2}\right) = -e^{-\b_2}\left(\frac{1+2\b_2}{4}\right) \frac{d \b_2}{d \a}~.$$
Rearranging, we get
$$ \frac{d \b_1}{d \a} = \frac12 \; e^{2\b_1-\a}\,\left(\frac{\a}{1+2\b_1}\right)  $$
and
$$ \frac{d \b_2}{d \a} = e^{\b_2-\a}\left(\frac{2+2\a}{1+2\b_2}\right) ~.$$
Let $(\a_0,\b_0)$ for $0<\a_0<1$ be an intersection of the curves, i.e., when $\a=\a_0$ we have $\b_1=\b_2=\b_0$. Therefore, $d\b_2/d\a>d\b_1/d\a$ if
	$$ e^{\b_2-\a}\left(\frac{2+2\a_0}{1+2\b_0}\right) > \frac12 \; e^{2\b_0-\a_0}\,\left(\frac{\a_0}{1+2\b_0}\right)~, $$
which yields
\begin{equation}\label{eq:beta-bound}
	4\,\left(1+\frac{1}{\a_0}\right) > e^{\b_0}~.
\end{equation}
Assume $\b_0>0$, otherwise Eq.~\eqref{eq:beta-bound} holds and we are done. By Eq.~\eqref{eq:alpha-beta3}, we get
	$$ e^{-\a_0}(1+\a_0) = e^{-2\b_0}(1+\b_0) \leq e^{-2\b_0}(1+2\b_0)~.$$
Since $ e^{-x}(1+x) $ is monotone decreasing for $x>0$, it follows that $2\b_0 \leq \a_0$. Hence, since $\a_0<1$, we get
	$$ e^{\b_0} \leq e^{\a_0/2} < e^{1/2} < 4\,\left(1+\frac{1}{\a_0}\right)~,$$
and thus Eq.~\eqref{eq:beta-bound} holds. This proves that the curves defined by Eq.~\eqref{eq:alpha-beta3} and \eqref{eq:alpha-beta4} intersect only once for $\a>0$.

It follows that Eq.~\eqref{eq:alpha-beta1} and \eqref{eq:alpha-beta2} are satisfied if and only if $\a\geq\a_0$. Hence, by Lemma~\ref{lem:alpha-beta} and Lemma~\ref{clm:reparam}, the canonical profile $\Xnr$ for $\r=\a(n+1)-2\ln((1+\a)/2)$ is a Nash equilibrium if and only if $\a\geq\a_0$. Since, by Observation~\ref{obs:monotone-lambda}, $\r$ is strictly increasing as a function of $\a$, the theorem follows.

The second part of the theorem is obtained by numerical approximation of the point of intersection $(\a_0,\b_0)$.
\qed

\commentt
{\bf DP: DO WE WANT TO INCLUDE THE FOLLOWING SECTION? OR IS IT ALREADY IN ARXIV?}
\section{Application: Fault Tolerance (OLD PAPER MODEL SECTION)}

In this section we present the location model studied in this paper.

\paragraph{Servers and Clients.}
The system consists of a finite set $N=\{1,\ldots,n\}$ of servers (acting as the \emph{players} in our game formulation), each of whom has to decide where to set shop along the interval $[0,1]$. We assume that clients are uniformly distributed along the line, and that they choose the closest server that is not disconnected from them. Each server wants to maximize its expected market share in the presence of faults. It is possible for more than one server to occupy the same location. In that case, clients choosing that location are divided equally between the servers at that location.

\paragraph{Faults.}
The model assumes that faults, or disconnections, occur at random along the line segment $[0,1]$, and clients cannot choose servers beyond a disconnection. The set of faults, denoted as $\F$, is distributed along the line at random according to a Poisson process with rate $\r>0$. More specifically, the probability that $k$ faults occur in any segment $[a,b] \subseteq [0,1]$ is
$$
    \Prb\left[|\F \cap [a,b]|=k\right]= \frac{e^{-\r(b-a)}\cdot (\r(b-a))^k}{k!}~.
$$
Namely, it is a Poisson distribution with rate $\r(b-a)$. Intuitively, the Poisson process could be viewed as a continuous analogue of the Binomial distribution. Divide the $[0,1]$ line into small segments of length $\delta>0$ and let $p=\frac{\r}{\delta}$. At each $\delta$ segment a fault occurs with probability $p \cdot \r$. As $\delta \to 0$ the number of faults in a segment $[a,b] \subseteq [0,1]$ converges to the Poisson distribution with rate $\r(b-a)$. Alternatively, we can define the Poisson process as a sequence of exponential random variables, i.e., the distance between each pair of consecutive faults is an exponential random variable with rate $\r$. (For a formal definition of the Poisson process see Appendix~\ref{sec:poisson-process}.)

\paragraph{Markets in the presence of faults.}
Consider a specific instance of the system. In this instance, let $\X = (x_1,\ldots,x_n) \subset [0,1]$ be the profile of the server locations. Let $ \F = \{f_1,\ldots, f_k\} $ be a given set of faults that have occurred. Two servers $i,j\in N$  are said to be \emph{colocated} if $x_i=x_j$. For $i\in N$, the set of $i$'s colocated servers is defined as $ \Gamma_i= \{j \in N \mid x_j=x_i\}$, and the size of this set is defined as $\gamma_i= |\Gamma_i|$. A server that is not colocated with other servers is called \emph{isolated}. Two servers are called \emph{neighbors} if no server is located strictly between them. A \emph{left (resp., right) peripheral server} is a server that has no servers to its left (resp., right). The servers divide the line into \emph{regions} of two types: \emph{internal regions}, which are regions between two neighbors, and two \emph{hinterlands}, which include the region between 0 and the left peripheral server, and the region between the right peripheral server and 1. (See Figure~\ref{fig:NE-shape}, where the two hinterlands are marked by $H$.)

The \emph{market} of each server $i\in N$ is the line segment in which clients choose location $x_i$. Note that colocated servers have the same market. Let $x_i^\ell$ and $x_i^r$ be the locations of $i$'s left and right neighbor, respectively. When no left (resp., right) neighbor exists we define $x_i^\ell=-1$ (resp., $x_i^r=2$). Let $f_i^\ell$ and $f_i^r$ be the closest faults to the left and right of $i$, respectively. When there are no faults to the left (resp., right) of $x_i$ we define $f_i^\ell=-\infty$ (resp., $f_i^r=\infty$). We define the market of $i\in N$ as the segment $[L_i,R_i]$, where $L_i$ and $R_i$ are defined as follows:
$$
\hbox{$
L_i=\left\{
				\begin{array}{ll}
					f_i^\ell,		& \mbox{if } f_i^\ell \geq x_i^\ell~; \\
                    0,               & \mbox{if } f_i^\ell < x_i^\ell=-1~; \\
					x_i^\ell, 	 & \mbox{otherwise.}
				\end{array}
			\right.
$}
~~~~~~~~~~~~~~
\hbox{$
	R_i=\left\{
				\begin{array}{ll}
					f_i^r,		& \mbox{if } f_i^r \leq x_i^r~; \\
                    1,          & \mbox{if } f_i^r > x_i^r=2~; \\
					x_i^r, 	 & \mbox{otherwise.}
				\end{array}
			\right.
$}
$$
In the definition of $L_i$,
the first case handles a situation where a failure occurs to the left of $i$ but to the right of its left neighbor if exists. The second concerns the case where there are neither neighbors nor failures to the left of $i$. The third handles a case where there is no failure between $i$ and its left neighbor.
The definition of $R_i$ is analogous.


Server $i$'s market is divided into two parts, $[L_i,x_i]$ and $[x_i,R_i]$, referred to as server $i$'s \emph{left} and \emph{right} half-markets, respectively.

\paragraph{The Game.}
The \emph{fault-prone Hotelling game} is denoted as $\fthotel(n,\F)$, where $N=\{1,\ldots,n\}$ is the set of players and $\F$ is the distribution of faults. For $i \in N$, the \emph{action} of player $i$ consists of selecting its location, $x_i \in [0,1]$ . The vector $ \X = (x_i)_{i\in N}$ is the profile of actions. Let $\Xm i$ denote the profile of actions of all the players different from $i$. Slightly abusing notation, we will denote by $(x'_i, \Xm i)$ the profile obtained from a profile $\X$ by replacing its $i$th coordinate $x_i$ with $x'_i$.

We denote the \emph{size} of server $i$'s market $[L_i,R_i]$ by $D_i = R_i- L_i$. This is itself a random variable, and we define the \emph{payoff} of player $i$ given the profile $\X$ as the expectation (over the possible failure configurations) of $D_i$ divided by the number of players colocated with $i$, namely,
$$u_i(\X)= \frac{\Ebb\left[ D_i \right]}{\gamma_i}~.$$
However, for the analysis, it is more convenient to view the payoff as composed of the \emph{left} and \emph{right payoffs}, $D_i=D_i^\ell+D_i^r$, where $D_i^\ell=x_i-L_i$ and $D_i^r=R_i-x_i$, and analyze $D_i^\ell$ and $D_i^r$ separately. The reason for this is that, as it turns out, $\Ebb[D_i^r]$ is only a function of the length $x_i^r-x_i$ of the right region of $i$ and the failure distribution, and similarly $\Ebb[D_i^\ell]$ depends only on $x_i-x_i^\ell$ and the failures, making it easier to analyze them separately.
\commenttend


\paragraph{Acknowledgments.}
The authors would like to thank
Shahar Dobzinski and Yinon Nahum for many fertile discussions and helpful
insights, and the anonymous reviewers for their useful comments.
This research was supported in part by a US-Israel BSF Grant No. 2016732.

\clearpage

\bibliographystyle{plain}
\bibliography{bib}

\end{document}